\documentclass{amsart}

\usepackage{xcolor}
\usepackage{xspace}
\usepackage[T1]{fontenc}
\usepackage[latin9]{inputenc}
\usepackage{geometry}
\usepackage{verbatim}
\usepackage{amstext}
\usepackage{amsthm}
\usepackage{amssymb}

\makeatletter

\theoremstyle{plain}
\newtheorem*{thm*}{Theorem}
\newtheorem{thm}{Theorem}[section]
\newtheorem{lem}[thm]{Lemma}
\newtheorem{cor}[thm]{Corollary}
\newtheorem{prop}[thm]{Proposition}
\newtheorem*{prop*}{Proposition}
\newtheorem*{lem*}{Lemma}

\theoremstyle{definition}
\newtheorem*{eg*}{Example}
\newtheorem*{egs*}{Examples}

\newtheorem*{Q*}{Question}

\theoremstyle{remark}
\newtheorem{rem}[thm]{Remark}

\allowdisplaybreaks

\DeclareMathOperator{\Tr}{Tr}

\usepackage{newtxtext}
\usepackage{dsfont}

\numberwithin{equation}{section}

\usepackage{graphicx,amssymb}

\makeatother

\usepackage[english]{babel}

\newcommand{\ii}{\mathrm{i}}
\newcommand{\cH}{\mathcal{H}}
\newcommand{\ud}{\mathrm{d}}
\usepackage{bbm}

\newcommand{\U}{\mathcal{U}}
\newcommand{\N}{\mathcal{N}}

\newcommand{\cL}{\mathcal{L}}
\newcommand{\cK}{\mathcal{K}}
\newcommand{\cU}{\mathcal{U}}
\newcommand{\cW}{\mathcal{W}}
\newcommand{\cF}{\mathcal{F}}
\newcommand{\cN}{\mathcal{N}}

\newcommand{\tr}{\mbox{Tr}}

\newcommand{\wt}{\widetilde}

\begin{document}
\title{On the characterisation of fragmented Bose-Einstein condensation\\ and its emergent effective evolution}

\author[Jinyeop Lee]{Jinyeop Lee}
\address[J.~Lee]{Mathematics Institute, Ludwig-Maximilian University of Munich, Theresienstrasse 39, Munich, 80333, Germany}
\email{lee@math.lmu.de}

\author[Alessandro Michelangeli]{Alessandro Michelangeli}
\address[A.~Michelangeli]{Institute for Applied Mathematics, and Hausdorff Center for Mathematics, University of Bonn, Endenicher Allee 60, D-53115 Bonn, GERMANY}
\email{michelangeli@iam.uni-bonn.de}

\maketitle

\begin{abstract}
		Fragmented Bose-Einstein condensates are large systems of identical bosons displaying \emph{multiple} macroscopic occupations of one-body states, in a suitable sense. The quest for an effective dynamics of the fragmented condensate at the leading order in the number of particles, in analogy to the much more controlled scenario for complete condensation in one single state, is deceptive both because characterising fragmentation solely in terms of reduced density matrices is unsatisfactory and ambiguous, and because as soon as the time evolution starts the rank of the reduced marginals generically passes from finite to infinite, which is a signature of a transfer of occupations on infinitely many more one-body states.
	In this work we review these difficulties, we refine previous characterisations of fragmented condensates in terms of marginals, and we provide a quantitative rate of convergence to the leading effective dynamics in the double limit of infinitely many particles and infinite energy gap.
\end{abstract}

\section{Introduction and background. BEC with fragmentation.}

 The goal of this work is to clarify and quantify certain emergent behaviours of the time evolution of large systems of identical bosons where Bose-Einstein condensation occurs in fragmented form. In this introductory Section we lay down the background for the analysis that is going to be be discussed in Sections \ref{sec:manybodyfrag} to \ref{sec:effectiveMFdynamics}.

 \subsection{Composite BEC}\label{sec:compositeBEC}~

 Since the very first realisations of Bose-Einstein condensation, a wide spectrum of experiments have been performed and studied where condensation occurs in a \emph{composite} form, as opposite to the ordinary \emph{simple} (i.e., almost complete and one-species) condensation. 
 
 In its simple form, Bose-Einstein condensation \cite{pethick02,pita-stringa-2016} (also BEC henceforth) is that inherently quantum phenomenon occurring in systems of a large number of identical bosonic particles at ultra-low temperature, and consisting of a macroscopic occupation of a single one-body state, thus with all particles displaying a common behaviour as if they were the same one: 
 apart from a (possibly, but not necessarily, very small) fraction of \emph{depletion}, the macroscopic occupation takes place in one single one-body state, with no internal of freedom. This is what ordinarily occurs with highly dilute and weakly interacting mono-atomic samples of alkali atoms, including the first two realisations of BEC with $^{87}$Rb \cite{Wieman-Cornel-Science-1995-BEC} and with $^{23}$Na \cite{Ketterle95} in 1995.

 On the other hand, composite condensation encompasses a variety of settings of non-simple BEC with internal structure of various sort. \emph{Condensate mixtures} \cite[Chapter 21]{pita-stringa-2016} consist of a gas formed by different species of interacting bosons, each of which is brought to condensation, thus with a macroscopic occupation of a one-body orbital for each species, and no inter-particle conversion. \emph{Quasi-spinor condensates} \cite{Ketterle_StamperKurn_SpinorBEC_LesHouches2001, Malomed2008_multicompBECtheory, Hall2008_multicompBEC_experiments,StamperKurn-Ueda_SpinorBose_Gases_2012} are gases of ultra-cold atoms that exhibit BEC and possess internal spin degrees of freedom which are often coupled to an external resonant micro-wave or radio-frequency radiation field, however, with no significant spin-spin internal interaction. In \emph{spinor condensates} (see the above references) the spin is an actual degree of freedom in interacting Bose gases of ultra-cold atoms where the spatial two-body interaction is mediated by a spin-spin coupling, and condensation manifests as a reversible spin-changing collisional coherence between particles. \emph{Fragmented condensates} \cite{Spekkens-Sipe-1998,Mueller-Ho-Ueda-Baym_PRA2006_fragmentation} are characterised by the occurrence of BEC as multiple macroscopic occupations of certain one-body states.

 The physical study of Bose-Einstein has triggered over the last two decades a voluminous corpus of mathematical investigations, through a multitude of different techniques, for the rigorous investigation of the ground state properties of Bose gases and the rigorous derivation of effective dynamical equations for the evolution of the condensate -- in view of the vastness of the subject, we refer to the monographs \cite{LSeSY-ober,Benedikter-Porta-Schlein-2015} and the references therein for a comprehensive discussion of \emph{simple} BEC, as well as to the latest improvements \cite{Pickl-RMP-2015,Benedikter-DeOliveira-Schlein_QuantitativeGP-2012_CPAM2015,Boccato-Cenatiempo-Schlein-2015_AHP2017_fluctuations,Chen-Lee-Lee-JMP2018-rateHartree,Lee-JSP2019-timedepRate,Bossman-Pavlovic-Pickl-Soffer-2020,BSS-2021,NNRT-2020}. Mathematical analyses of various types of composite condensations were recently produced for mixtures \cite{M-Olg-2016_2mixtureMF,AO-GPmixture-2016volume,Anap-Hott-Hundertmark-2017,DeOliveira-Michelangeli-2016,MNO-2017,Michelangeli-Pitton-2018,LeeJ-mixt2020-JMP2021,MS-2021-GrowthCorrMixt}, quasi-spinor condensates \cite{MO-pseudospinors-2017}, spinor condensates \cite{MO-2018-spin-spin}, and fragmented condensates \cite{DimFalcOlg21-fragmented}.

  \subsection{Simple vs fragmented BEC}~
  
  Despite the fact that fragmented BEC is the object of intensive physical study, experimental and theoretical \cite{Spekkens-Sipe-1998,Alon-Streltsov-Cederbaum_PLA2005_fragmentation,Alon-Streltsov-Cederbaum-Lorenz-PLA2006_fragmentation,Mueller-Ho-Ueda-Baym_PRA2006_fragmentation,Sakmann-etAl-2008-fragm,Bader-Fischer-2009,Fischer-Bader-PRA2010,Mullin-Sakhel-2012,Kang-Fischer-PRL2014,Fischer-Lode-Chatterjee-2015,Lode-PRA2016,Tomchenko-2019,Lode-Dutta-Leveque-2021}, one soon realises that its customary definition as a `macroscopic occupation of two or more one-body states' \cite[Section III.A]{legget} is deceptively simple on mathematical grounds. To discuss this (Section \ref{sec:fragmentation-marginals} below), let us preliminary revisit (here and in Section \ref{sec:asymptotics}) the mathematical formalisation of simple BEC and the naive generalisation to fragmented BEC.

 A pure state of a system of $N$ identical bosons in $d$ dimensions is described by a unit vector $\Psi_N$ belonging to a Hilbert space $\cH_N$ that has the form $\cH_N=\cH^{\otimes_{\mathrm{sym}} N}$, the symmetric tensor product of $N$ copies of the same single-particle Hilbert space $\cH$, or more generally, if the state is non-pure, by a density matrix (i.e., a normalised, positive, self-adjoint operator) $\gamma_N$ acting on $\cH_N$ (for a pure state, $\gamma_N=|\Psi_N\rangle\langle\Psi_N|$). When in practice the particle spin does not participate in the inter-particle interaction and hence effectively $\cH=L^2(\mathbb{R}^d)$, the bosonic symmetry manifests as the invariance of $\Psi_N(x_1,\dots,x_N)$ under exchange $x_j\leftrightarrow x_k$ of any pairs of variables $x_j,x_k\in\mathbb{R}^d$. For a generically mixed state, bosonic symmetry is tantamount as the invariance of the density matrix's integral kernel $\gamma_N(x_1,\dots,x_N;y_1,\dots,y_N)$ under simultaneous permutation of the $j$-th and $k$-th variable in each of the two sets of variables of the kernel (which is in fact equivalent, owing to the self-adjointness of $\gamma_N$, to the sole invariance of $\gamma_N(x_1,\dots,x_N;y_1,\dots,y_N)$ under permutation of $x$-variables only, or of $y$-variables only).

 To each $N$-body state one naturally associates the notion of \emph{occupation number}, intuitively speaking the fraction of the $N$ particles occupying the same one-body state, in the usual sense of reduced density matrices.

 To this aim, the operation of $k$-th body \emph{partial trace} is introduced, for a fixed $k\in\{1,\dots,N-1\}$, as the map
 \begin{equation}\label{eq:pt1}
  \begin{split}
   \mathcal{L}^1(\cH^{\otimes_{\mathrm{sym}} N})\;&\longrightarrow \mathcal{L}^1(\cH^{\otimes_{\mathrm{sym}} k})\,, \\
   T\;&\longmapsto\;\mathrm{Tr}_{[N-k]}(T)
  \end{split}
 \end{equation}
 between trace-class operators on $\cH^{\otimes_{\mathrm{sym}} N}$ and trace-class operators on $\cH^{\otimes_{\mathrm{sym}} k}$ defined by
 \begin{equation}\label{eq:PT-def1}
  \big\langle \varphi,\mathrm{Tr}_{[N-k]}(T)\psi\big\rangle_{\cH^{\otimes_{\mathrm{sym}} k}}\;=\;\sum_j \big\langle \varphi\otimes\xi_j,T\,\psi\otimes\xi_j\big\rangle_{\cH^{\otimes_{\mathrm{sym}} N}}\qquad\forall\,\varphi,\psi\in\cH^{\otimes_{\mathrm{sym}} k}\,,
 \end{equation}
 where $(\xi_j)_j$ is an orthonormal basis of $\cH^{\otimes_{\mathrm{sym}}(N-k)}$, having assumed that the one-body Hilbert space is separable, and with the customary notation $\langle\cdot,\cdot\rangle$ for the Hilbert space scalar product (with explicit indication of the underlying Hilbert space, when needed), conventionally anti-linear in the first entry and linear in the second. \emph{Equivalently}, $\mathrm{Tr}_{[N-k]}(T)$ is characterised by
 \begin{equation}\label{eq:PT-def2}
  \mathrm{Tr}\big( A\,\mathrm{Tr}_{[N-k]}(T)\big)\;=\; \mathrm{Tr}\big( (A\otimes\mathbbm{1}_{N-k}) T \big)\qquad\forall A\in\mathcal{B}(\cH^{\otimes_{\mathrm{sym}} k})\,,
 \end{equation}
 $\mathbbm{1}_{N-k}$ being the identity operator on $\cH^{\otimes_{\mathrm{sym}}(N-k)}$, $\mathcal{B}(\cdot)$ denoting the everywhere defined and bounded operators on the considered Hilbert space, and $\mathrm{Tr}(\cdot)$ denoting the trace of the considered operator, omitting for shortness the declaration of the Hilbert space that operator acts on. In fact, the partial trace preserves the trace and the positivity, hence it maps $N$-body into $k$-body density matrices: the operator
 \begin{equation}\label{eq:pt4}
  \gamma_N^{(k)}\;:=\;\mathrm{Tr}_{[N-k]}(\gamma_N)
 \end{equation}
 is called \emph{$k$-body reduced density matrices} (or \emph{$k$-marginal}) associated to $\gamma_N$.  
 One interprets \eqref{eq:PT-def1} or \eqref{eq:PT-def2} by saying that the partial trace of $\gamma_N$ amounts to tracing out all but $k$ degrees of freedom: in order to evaluate the expectation in the state $\gamma_N$ of an observable that acts non-trivially as $A$ on $k$ particles only, it suffices to know $\gamma_N^{(k)}$, as the desired expectation is equal to the expectation of $A$ on the $k$-body state $\gamma_N^{(k)}$. When $\cH=L^2(\mathbb{R}^d)$, the $k$-marginal's kernel takes the familiar form
 \begin{equation}\label{eq:tracing-out}
   \gamma_N^{(k)}(x_1,\dots,x_k;y_1,\dots,y_k)\;=\;\int_{\mathbb{R}^{3(N-k)}} \gamma_N(x_1,\dots,x_k,Z;y_1,\dots,y_k,Z)\,\ud Z
 \end{equation}
 for variables $x_1,\dots,x_k,y_1,\dots,y_k\in\mathbb{R}^d$.

 The one-body marginal $\gamma_N^{(1)}$ thus encodes the necessary information for evaluating expectations of one-body observables in the considered $N$-body state $\gamma_N$. As such, the canonical singular value decomposition of $\gamma_N^{(1)}$, namely
 \begin{equation}\label{eq:gammaN1exp}
  \gamma_N^{(1)}\;=\;\sum_{j=0}^\infty \,n_j^{(N)}|\varphi_j^{(N)}\rangle\langle\varphi_j^{(N)}|\,,\qquad 1\geqslant n_0^{(N)}\geqslant n_1^{(N)} \geqslant\cdots\geqslant 0\,,\quad \sum_{j=0}^\infty \,n_j^{(N)}=1
 \end{equation}
 for some surely existing orthonormal basis $(\varphi_j^{(N)})_{j=0}^\infty$ of $\cH$, leads to interpret each coefficient $n_j^{(N)}$ as the fraction, or occupation, of the total number of particles occupying the one-body states $\varphi_j^{(N)}$, clearly in the sense of reduced marginals, since the actual many-body state does not display in general a rigorous factorisation into one-body states.

 According to the standard Landau-Penrose-Onsager definition \cite{Landau-Lifshitz-5,penrose-1951,PO-1956} (see also \cite[Section III.A]{legget}) in terms of occupation numbers, physicists refer to \emph{simple BEC} as the occurrence where $n_0^{(N)}=O(1)$ and $n_j^{(N)}=o(1)$ for $j\neq 0$ in the large parameter $N$: only one single one-body state is macroscopically occupied, the occupation numbers of all other one-body states being negligible. They speak instead of \emph{fragmented BEC} when there are $\ell\geqslant 2$ one-body states that are macroscopically occupied, namely $n_0^{(N)},\dots,n_{\ell-1}^{(N)}=O(1)$ and $n_j^{(N)}=o(1)$ for $j\in\{\ell,\ell+1,\ell+2,\dots\}$. The fraction $1-n_0$ (for simple BEC) or $1-\sum_{j=0}^{\ell-1}n_j$ (for fragmented BEC)  of non-condensed particles expresses the depletion of the system. Depletion-less BEC is said to be \emph{complete} (or 100\%).

 Mathematically this is customarily monitored in some kind of rigorous limit of infinitely many particles, which replaces the physical notion of `large' $N$: thus, (asymptotically) simple BEC would correspond to
 \begin{equation}\label{eq:defSimpleBEC}
  n_0^{(N)}\xrightarrow[]{\,N\to\infty\,}\,n_0\,\in\,(0,1)\,,\qquad n_j^{(N)}\xrightarrow[]{\,N\to\infty\,}\,0\,,\; j\neq 0\,\qquad (\textsc{simple bec})\,,
 \end{equation}
 and (asymptotically) fragmented BEC would correspond to
  \begin{equation}\label{eq:defFragmBEC}
  \begin{split}
   n_j^{(N)}&\xrightarrow[]{\,N\to\infty\,}\,n_j\in(0,1)\qquad \textrm{for }j\in\{0,\dots,\ell-1\}\,, \\
   n_j^{(N)}&\xrightarrow[]{\,N\to\infty\,}\,0\qquad\qquad\qquad \textrm{for }j\in\{\ell,\ell+1,\ell+2,\dots\}
  \end{split}\qquad\qquad (\textsc{fragmented bec})\,.
 \end{equation}

 \subsection{Asymptotic definitions}\label{sec:asymptotics}~
 
 The usage of the limit $N\to\infty$ in \eqref{eq:defSimpleBEC}-\eqref{eq:defFragmBEC} encompasses a slightly excessive generality that one customarily restricts by adopting an equally familiar but stronger definition of BEC than the above physical jargon of $O(1)$-size of one or more occupation numbers. Indeed, whereas at large, but \emph{fixed} $N$ (the actual number of particles in an experiment with Bose-Einstein condensates), having all the $n_j^{(N)}$'s negligible for $j\neq 0$ and the sole $n_0^{(N)}\sim n_0$ of finite (non-zero) size, say, $n_0=1$, is immediately interpreted by saying that the sole one-body state $\varphi_0^{(N)}$ is macroscopically occupied, and $\gamma_N^{(1)}\approx |\varphi_0^{(N)}\rangle\langle\varphi_0^{(N)}|$, instead as $N\to\infty$ \eqref{eq:defSimpleBEC} does not necessarily imply that $\gamma_N^{(1)}\to|\varphi\rangle\langle\varphi|$ for some $\varphi\in\cH$. (Example: $\cH=\ell^2(\mathbb{N})$ and $\gamma_N^{(1)}=(1-\frac{1}{N})|e_N\rangle\langle e_N|+\frac{1}{N}|e_1\rangle\langle e_1|$, where $(e_j)_{j\in\mathbb{N}}$ is the canonical basis of $\ell^2(\mathbb{N})$.)

 For this reason the mathematically convenient and physically meaningful definition of 100\% simple BEC onto the one-body state $\varphi$, asymptotically in $N$, is rather
  \begin{equation}\label{eq:defSimpleBEC-2}
   \gamma_N^{(1)}\xrightarrow[]{\,N\to\infty\,}\,|\varphi\rangle\langle\varphi|\,,\qquad\textrm{equivalently,}\qquad \langle\varphi,\gamma_N^{(1)}\varphi\rangle_\cH\xrightarrow[]{\,N\to\infty\,}\,1 \qquad (\textsc{100\% simple bec}).
  \end{equation}
 This definition, stronger than \eqref{eq:defSimpleBEC}, retains the physical interpretation: $\varphi$ is a normalised vector in $\cH$ that represents the condensate state, namely the state of complete occupation of the assembly of bosons, and is customarily referred to as the \emph{order parameter} or \emph{one-body orbital} of the condensate. The above asymptotic vicinity of the one-body marginal to the rank-one projection onto $\varphi$ can be equivalently formulated in \emph{any} topology ranging from the weak operator to the trace norm convergence: indeed \cite{am_equivalentBEC,kp-2009-cmp2010},
 \begin{equation}\label{eq:equivalent-BEC-control-1compontent}
1-\langle\varphi,\gamma_N^{(1)}\varphi\rangle\;\leqslant\;\mathrm{Tr}\big|\,\gamma_{N}^{(1)}-|\varphi\rangle\langle\varphi|\,\big|\;\leqslant\;2\sqrt{1-\langle\varphi,\gamma_N^{(1)}\varphi\rangle}\,.
\end{equation}
 Moreover, asymptotic 100\% simple BEC can be equivalently characterised in the above asymptotic sense at the level of any $k$-body marginal: indeed \cite{LSe-2002,am_equivalentBEC,kp-2009-cmp2010},
 \begin{equation}\label{eq:1klevel}
  1-\langle\varphi,\gamma_N^{(1)}\varphi\rangle_{\cH}\;\leqslant\;1-\langle\varphi^{\otimes k},\gamma_N^{(k)}\varphi^{\otimes k}\rangle_{\cH^{\otimes k}}\;\leqslant\;k\big(1-\langle\varphi,\gamma_N^{(1)}\varphi\rangle_{\cH}\big)\,,\qquad k\in\{1,\dots,N\}\,.
 \end{equation}
 Thus, 100\% simple BEC is equivalent to the asymptotic factorisation $\gamma_N^{(k)}\approx|\varphi^{\otimes k}\rangle\langle\varphi^{\otimes k}|$ as $N\to\infty$, at any fixed level $k$. All this encodes the informal idea that the many-body state of complete simple condensation is a vector $\Psi_N\in\cH_N$ essentially of the form $\varphi^{\otimes N}$, but such a factorisation only makes sense at the level of reduced density matrices and in general the difference $\Psi_N-\varphi^{\otimes N}$ does not vanish at all in $\cH_N$ as $N\to\infty$ (a possibly complicated pattern of correlations is present in $\Psi_N$ that is not detected at the level of marginals \cite{Lewin-Nam-Serfaty-Solovej-2012_Bogolubov_Spectrum_Interacting_Bose}).

 When one comes to formalise \emph{fragmented condensation} asymptotically in $N$, say, for simplicity and with no essential loss of generality, with zero depletion, the first natural analogue to \eqref{eq:defSimpleBEC-2}, in view of \eqref{eq:gammaN1exp} and \eqref{eq:defFragmBEC}, is
 \begin{equation}\label{eq:defFragmBEC-2}
  \gamma_N^{(1)}\xrightarrow[]{\,N\to\infty\,}\,\sum_{j=0}^{\ell-1}n_j|\varphi_j\rangle\langle\varphi_j|\qquad (\textsc{100\% $\ell$-level fragmented bec})
 \end{equation}
 for some orthonormal vectors (multiple one-body orbitals) $\varphi_0,\dots\varphi_{\ell-1}\in\cH$ and non-zero weights (occupation numbers) $n_0,\dots,n_{\ell-1}\in(0,1)$ summing up to 1. Again, \eqref{eq:defFragmBEC-2} is stronger than \eqref{eq:defFragmBEC} precisely as \eqref{eq:defSimpleBEC-2} is stronger than \eqref{eq:defSimpleBEC}. Fragmentation thus emerges as the asymptotic rank-$\ell$ one-body marginal for $\ell\geqslant 2$.

 \subsection{Fragmentation at the many-body and at the marginals level}\label{sec:fragmentation-marginals}~

 Meaning fragmentation as in \eqref{eq:defFragmBEC-2} (or \eqref{eq:defFragmBEC}) is directly inspired from physical reasoning, but it is common sense that it encompasses a too broad variety of states. Consider, for instance, the two many-body bosonic states
 \begin{equation}\label{eq:twoexamples}
  \Xi_N\,:=\,\varphi_{1}^{\otimes {N/2}}\vee\varphi_{2}^{\otimes {N/2}}\,,\qquad \Upsilon_N\,:=\,\frac{1}{\sqrt{2}}\big(\varphi_1^{\otimes N}+\varphi_2^{\otimes N}\big)\,,
 \end{equation}
 where  $\varphi_1,\varphi_2\in\cH$, $\langle\varphi_m,\varphi_n\rangle_{\cH}=\delta_{mn}$, $N\in 2\mathbb{N}$, and where `$\vee$' is the customary notation for the overall symmetric tensor product between the two factors. 
 The two $N$-body vectors above have the same rank-2 one-body marginal
 \begin{equation}\label{eq:stupmarg}
  \frac{1}{2}\,|\varphi_1\rangle\langle\varphi_1|+\frac{1}{2}\,|\varphi_2\rangle\langle\varphi_2|\,,
 \end{equation}
 yet $\Xi_N$ and $\Upsilon_N$ are significantly different in nature: $\Upsilon_N$ is a mere superposition of two uncorrelated $N$-body condensates with 100\% simple BEC respectively on $\varphi_1$ and $\varphi_2$, whereas $\Xi_N$ correlates an equal number $N/2$ of identical $\varphi_1$-bosons and identical $\varphi_2$-bosons with the precise amount of correlations dictated by the overall bosonic statistic. Despite both satisfying \eqref{eq:defFragmBEC-2}, only $\Xi_N$ supports a meaningful interpretation of fragmented BEC. \emph{This is consistent with the customary preparation of fragmented condensates in experiments}, where the two one-body states $\varphi_1,\varphi_2$ macroscopically occupied by the $N$ identical bosons may be ground states of spatially well separated traps, or energetically well separated hyperfine levels, or the like.

 A related signature of the excessive broadness of definition \eqref{eq:defFragmBEC-2} is the lack of a counterpart for higher order $k$-body marginals, like a control such as \eqref{eq:1klevel}. Asymptotically rank-$1$ one-body marginals necessarily imply asymptotically rank-$1$ $k$-body marginals, but this ceases to be true for rank-$\ell$ one-body marginals when $\ell\geqslant 2$. In \eqref{eq:twoexamples}, $\Upsilon_N$ has rank-2 $k$-body marginals
 \[
  \gamma_{\Upsilon_N}^{(k)}\;=\;\frac{1}{2}\big( |\varphi_1^{\otimes k}\rangle\langle \varphi_1^{\otimes k}| + |\varphi_2^{\otimes k}\rangle\langle \varphi_2^{\otimes k}| \big)
 \]
 for every $k\in\{1,\dots,N-1\}$, as is immediate to see by applying \eqref{eq:tracing-out} to
 \[
  \gamma_{\Upsilon_N}\;=\;|\Upsilon_N\rangle\langle\Upsilon_N|\;=\;\frac{1}{2}\Big(\big|\varphi_1^{\otimes N}\big\rangle\big\langle \varphi_1^{\otimes N}\big|+\big|\varphi_1^{\otimes N}\big\rangle\big\langle \varphi_2^{\otimes N}\big|+\big|\varphi_2^{\otimes N}\big\rangle\big\langle \varphi_1^{\otimes N}\big|+\big|\varphi_2^{\otimes N}\big\rangle\big\langle \varphi_2^{\otimes N}\big|\Big)
 \]
 and exploiting the orthogonality $\varphi_1\perp\varphi_2$. Instead $\Xi_N$ above has $k$-body marginals whose rank is finite and increases with $k$ (see \eqref{eq:gammak} and \eqref{eq:rankorig} below).

 The excessive scope of the elementary definition \eqref{eq:defFragmBEC-2} (or \eqref{eq:defFragmBEC}) has been known since long \cite{Pethick-Pitaevskii-2000}, and attempts of various sort have been made to characterise fragmentation in BEC in terms of the reduced density matrices of the system, aware of the above-mentioned difficulties. A typical idea is to restrict such a notion to states that, in the above language, display a finite and $\geqslant 2$ rank at the level of each $k$-body marginal. 
 This is implicitly built in in the context of approximate numerical methods for many-body time evolution within the MCTDHB scheme (multi-configurational time-dependent Hartree for bosons) \cite{Alon-Streltsov-Cederbaum_PLA2005_fragmentation,Alon-Streltsov-Cederbaum-Lorenz-PLA2006_fragmentation,Alon-Streltsov-Cederbaum-PRA2008fragm,Lode-PRA2016,Lode-Dutta-Leveque-2021}, and yet it does not factor out a state like $\Upsilon_N$ in \eqref{eq:twoexamples}. A recent refinement was proposed in \cite{DimFalcOlg21-fragmented} to mathematically characterise fragmented BEC as the occurrence of $k$-body marginals all with finite and non-constant-in-$k$ (in practice: increasing) rank, with also $\mathrm{rank}(\gamma_N^{(1)})\geqslant 2$. Upon inspection (Proposition \ref{prop:reduced-marginals-time-zero}), this definition too leaves room to an excess of generality, encompassing many-body states that do \emph{not} correspond to feasible preparations of fragmented condensates, even if it succeeds in including the fragmented state $\Xi_N$ and excluding the manifestly non-fragmented state $\Upsilon_N$ considered in \eqref{eq:twoexamples} (as well as excluding other similar, non-fragmented states like \cite[Eq.~(2.6)]{DimFalcOlg21-fragmented}).

 This discussion shows that, apart from the intuitive idea of multiple macroscopic occupations, \emph{the notion of fragmented BEC is not so naturally and consistently definable at the level of marginals}, as opposite to the notion of simple BEC. This has ultimately to do with the precise amount of inter-particle correlations that one actually admits as \emph{fragmentation}. Out of the states \eqref{eq:twoexamples},  $\Upsilon_N$ has no correlations, at any $k$-particle level, 
 whereas $\Xi_N$ has many, yet only those imposed by the overall bosonic symmetry.

 \subsection{Many-body state fragmentation characterised in terms of its preparation}\label{sec:fragmented-manybody}~
 
 It is therefore safer, for the purposes of the present analysis, to mathematically formalise fragmented condensates with reference to the actual physical preparation of certain many-body states. This brings our focus on a class of vectors of $\cH_N$ that are \emph{close in norm} to the prototypical, exactly fragmented states of the form
 \begin{equation}\label{eq:genfrag}
  \begin{split}
     & \Psi^{(\varphi_1,N_1;\dots;\varphi_\ell,N_\ell)}_N\;\equiv\;|\varphi_1,N_1;\dots;\varphi_\ell,N_\ell\rangle\;:=\;\varphi_1^{\otimes N_1}\vee\cdots\vee\varphi_\ell^{\otimes N_\ell}\,, \\
     & \qquad N_1,\dots,N_\ell\in\mathbb{N}\,,\quad N_1+\cdots+N_\ell=N\,, \quad \frac{N_j}{N}\,\xrightarrow[]{\,N\to\infty\,}\,n_j\in(0,1)\,,\;j\in\{1,\dots,\ell\}\,,\\
     & \qquad \varphi_1,\dots,\varphi_\ell\in\cH\,,\quad \langle\varphi_m,\varphi_n\rangle_\cH\,=\,\delta_{nm}
  \end{split}
 \end{equation}
  for some $\ell\in\mathbb{N}$, $\ell\geqslant 2$. A many-body state of type \eqref{eq:genfrag} unambiguously describes fragmented condensation of an assembly of indistinguishable bosons with asymptotic occupation numbers $n_1,\dots,n_\ell$ respectively onto the one-body states $\varphi_1,\dots,\varphi_\ell$. 
  The computation of the corresponding marginals (see, e.g., \cite[Proposition 2.6]{DimFalcOlg21-fragmented}, as well as Proposition \ref{prop:reduced-marginals-time-zero} and Remark \ref{rem:general-ell} below) shows that the $k$-marginal of \eqref{eq:genfrag} has the form
  \begin{equation}\label{eq:combkmarg}
   \sum_{\substack{ \mathbf{a}\equiv(a_1,\dots,a_\ell) \\ a_j\in\{1,\dots,N_j\} \\ a_1+\cdots+a_\ell= k \\  }}\binom{N}{k}^{-1}\binom{N_1}{a_1}\cdots\binom{N_\ell}{a_\ell}\,\big|\varphi_1^{\otimes a_1}\vee\cdots\vee\varphi_\ell^{\otimes a_\ell}\big\rangle\big\langle\varphi_1^{\otimes a_1}\vee\cdots\vee\varphi_\ell^{\otimes a_\ell}\big|
  \end{equation}
  (hence with \emph{finite} rank increasing with $k$), and in particular the $1$-marginal of \eqref{eq:genfrag} is
  \begin{equation}\label{eq:g1frag}
   \sum_{j=1}^\ell \frac{N_j}{N}|\varphi_j\rangle\langle\varphi_j|\;\xrightarrow[]{\,N\to\infty\,}\;\sum_{j=1}^\ell n_j|\varphi_j\rangle\langle\varphi_j|\,,
  \end{equation}
  that is, analogously to \eqref{eq:stupmarg}, 
  the same rank-$\ell$ density matrix obtained as $1$-marginal of the (clearly non-fragmented) many-body state $\sum_{j=1}^\ell\sqrt{\frac{N_j}{N}}\,\varphi_j^{\otimes N}$.

  In \eqref{eq:genfrag} condensation is exactly depletion-less and the sole amount of inter-particle correlations is the one dictated by the bosonic symmetry. In a state close in norm to a state of type \eqref{eq:genfrag} a small amount of depletion is possible, as well as a pattern of additional correlations. Recall indeed that for any two vectors $\Theta_N,\Lambda_N\in\cH_N$
  \begin{equation}\label{eq:vicinities}
   \mathrm{Tr}\big|\gamma_{\Theta_N}^{(k)}-\gamma_{\Lambda_N}^{(k)}\big|\;\leqslant\;2\,\|\Theta_N-\Lambda_N\|_{\cH_N}\qquad \forall k\in\{1,\dots,N\}\,.
  \end{equation}
 Thus, any many-body state which is $\varepsilon$-close in norm to \eqref{eq:genfrag}, having $k$-body marginals that are corrections to \eqref{eq:combkmarg}
 in trace norm, has still a natural interpretation of fragmentation onto the levels $\varphi_1,\dots,\varphi_\ell$, even if the marginals have possibly \emph{infinite} rank.

 The latter consideration also clarifies that, unlike previous attempts \cite{DimFalcOlg21-fragmented}, there is nothing conceptually stringent in requiring that fragmented BEC be `finite' (meaning: of finite rank at the level of each marginal): finite, $\varepsilon$-small (hence irrelevant) occupations of an infinite amount of one-particle states produce an infinite rank in $\gamma_N^{(1)}$, and yet $\gamma_N^{(1)}$ may be $\varepsilon$-close to \eqref{eq:g1frag}.
 
 The preparation of a bosonic system into a state of fragmented condensation approximately close to \eqref{eq:genfrag} involves a confinement of the sample so as to make the bosons to macroscopically occupy two or more single-particle states (wells of a spatial trap, hyperfine levels, and the like) which are mutually orthogonal and separated energetically. Intuitively speaking, the larger the separation, the less correlated the many-body state, namely the closer it is to \eqref{eq:genfrag}. When the level separation is too small, and the other experimental conditions for condensation are still matched, simple BEC (typically: in the ground state of the trap) becomes more favourable than fragmentation.

 \subsection{Fragmentation with finite or with infinite gap}\label{sec:frag-finInfinGap}~
 
 The latter consideration suggests that in certain aspects of the rigorous analysis of fragmentation, it is convenient to introduce one further asymptotic parameter, beside the number $N$ of particles, namely the gap $\nu\geqslant 0$ (an energy gap $\hbar\nu$, in units $\hbar=1$) at the bottom of the spectrum of the one-particle Hamiltonian $\mathfrak{h}_\nu$ acting on $\cH$, with the assumption that 
 \begin{equation}\label{eq:gapfirst}
  \nu\;=\;\inf\big(\sigma(\mathfrak{h}_\nu)\setminus\inf\sigma(\mathfrak{h}_\nu)\big)-\inf\sigma(\mathfrak{h}_\nu)
 \end{equation}
 ($\sigma(\mathfrak{h}_\nu)$ denoting the spectrum of $h_\nu$),
 assuming tacitly that $\inf\sigma(\mathfrak{h}_\nu)\geqslant C\in\mathbb{R}$ uniformly in $\nu$ (e.g.,  $\inf\sigma(\mathfrak{h}_\nu)=0$ $\forall\nu\geqslant 0$). In the standard setting $\cH=L^2(\mathbb{R}^d)$,  $\mathfrak{h}_\nu$ has typically the form $-\Delta+U_\nu$ (in units $\hbar=2m=1$), or variants of it, for some real-valued confining potential $U_\nu$. A typical many-body Hamiltonian is then
 \begin{equation}\label{eq:genericHam}
  \sum_{j=1}^N \mathfrak{h}_{\nu,j}\,+\!\!\sum_{1\leqslant j<r\leqslant N }\!\!V(x_j-x_r)\,,
 \end{equation}
 where $x_1,\dots,x_N\in\mathbb{R}^d$, $V$ is a real-valued two-body interaction potential, and $\mathfrak{h}_{\nu,j}$ acts non-trivially (i.e., other than the identity) as $\mathfrak{h}_\nu$ on the $j$-th variable (for instance, $\mathfrak{h}_{\nu,j}=-\Delta_{x_j}+U_\nu(x_j)$). The preparation of fragmented BEC in a Bose system governed by a Hamiltonian of the type \eqref{eq:genericHam} requires, as said, a sufficiently large gap $\nu$. This allows for the investigation of both the meaningful regimes $N\to\infty$, $\nu>0$ (finite and large gap), and $N\to\infty$, $\nu\to\infty$ (infinite gap).

 Interacting bosons in a deep, suitable scaled double-well potential were recently studied in \cite{Rougerie-Spehner-2017}, and in several references therein, where in the limit of asymptotically large spatial separation of the two wells the 50\%-50\% occupation of a left-localised one-particle state $\varphi_L$ and a right-localised one-particle state $\varphi_R$, corresponding to a many-body state $\varphi_L^{\otimes N/2}\vee\varphi_R^{\otimes N/2}$ as $N\to\infty$ was proved to be energetically favoured as compared to the regime of simple occupation of the delocalised one-body state $\frac{1}{\sqrt{2}}(\varphi_L+\varphi_R)$, expected when the two wells are spatially close, and corresponding to a many-body state $(\frac{\varphi_L+\varphi_R}{\sqrt{2}})^{\otimes N}$.

  \subsection{Evolution of fragmentation in a regime of finite or of infinite gap}\label{sec:Evolutiongeneral}~

 One further natural and fundamental question concerns the \emph{time evolution} of a fragmented state like (or close to) \eqref{eq:genfrag}, governed by a Hamiltonian of type \eqref{eq:genericHam}, in either alternatives of finite or infinite gap $\nu$, as $N\to\infty$. This mirrors the much more deeply understood counterpart problem of the persistence in time of simple BEC, namely from $\gamma_N^{(1)}\approx|\varphi\rangle\langle\varphi|$ to $\gamma_{N,t}^{(1)}\approx|\varphi_t\rangle\langle\varphi_t|$ at later times $t\geqslant 0$ and large $N$, and the rigorous derivation of the effective (non-linear) evolution equation for $\varphi_t$ (see \cite{BEGMY-2002,EY-2001,ESY-2006,ESYinvent,RS-2007,ESY-2008,kp-2009-cmp2010,Pickl-JSP-2010,Pickl-LMP-2011,chen-lee-2010-JMP2011,Chen-Lee-Schlein-2011,Pickl-RMP-2015,Benedikter-DeOliveira-Schlein_QuantitativeGP-2012_CPAM2015,Benedikter-Porta-Schlein-2015,Boccato-Cenatiempo-Schlein-2015_AHP2017_fluctuations,Chen-Lee-Lee-JMP2018-rateHartree,Lee-JSP2019-timedepRate,Bossman-Pavlovic-Pickl-Soffer-2020,NNRT-2020} and the references therein), and similarly the persistence in time of mixture/quasi-spinor/spinor BEC \cite{M-Olg-2016_2mixtureMF,AO-GPmixture-2016volume,Anap-Hott-Hundertmark-2017,DeOliveira-Michelangeli-2016,MO-pseudospinors-2017,MNO-2017,MO-2018-spin-spin,Michelangeli-Pitton-2018,LeeJ-mixt2020-JMP2021,MS-2021-GrowthCorrMixt}. Fragmentation brings in an amount of difficulties and additional levels for such a question.

 Concretely speaking, one investigates the Schr\"{o}dinger-evoluted state $\Psi_{N,\nu,t}=e^{-\ii t H_{N,\nu}}\Psi_{N,0}$, at times $t\geqslant 0$, of an initial state $\Psi_{N,0}$ of fragmented BEC, approximately of the form \eqref{eq:genfrag}, governed by a many-body Hamiltonian $H_{N,\nu}$ of the type \eqref{eq:genericHam}, suitably re-scaled so as to control the size $\nu$ of the one-body gap and the $N$-dependence of the inter-particle interaction. The re-scaling is an actual caricature of a genuine, yet non-controllable at present, thermodynamic limit, and is needed to keep physical meaningfulness of the picture as $N\to\infty$ or $\nu\to\infty$: in particular \cite{am_GPlim,Benedikter-Porta-Schlein-2015}, the interaction potential $V$ is re-scaled as $V_N$ in order to formally re-size the formal $O(N^2)$-contribution of the potential terms of $H_{N,\nu}$ to the actual formal $O(N)$-contribution of the kinetic terms, and further, possibly, in order to mimic various physically realistic regimes of high dilution, weak interaction, short effective range, etc.

 Based on physical heuristics one informally expects that in a large-but-finite gap regime fragmentation in (or close to) the form \eqref{eq:genfrag} undergoes a degree of deterioration due to the dynamical emergence of further inter-particle correlations -- e.g., tunnelling between wells in the spatial confinement of the sample -- which yields at later times a non-zero, although typically still very small, occupation of an infinity of further one-body states, hence infinite-rank reduced $k$-marginals at every order. Persistence of finite-rank fragmentation is instead expected in the opposite asymptotic regime of infinite gap.

 Identifying \emph{effective evolution equations} for the multiple order parameters of the fragmentation is problematic, and so too is a \emph{quantitative} control of the rate of convergence in $N$ and $\nu$ of the evolved marginals $\gamma_{N,t}^{(k)}$ to their asymptotic version. Both such difficulties are ultimately related to the large amount of statistical correlations present in a many-body state with fragmented BEC.

 Within the above-mentioned MCTDHB scheme one imposes the persistence of fragmented BEC at later times with a fixed number of one-body orbitals, thus prescribing $\mathrm{rank}(\gamma_{N,\nu,t}^{(1)})=\mathrm{rank}(\gamma_{N}^{(1)})$, thereby formally deriving a self-consistent system of coupled non-linear Schr\"{o}dinger equations, each for one of the considered orbitals \cite{Alon-Streltsov-Cederbaum-Lorenz-PLA2006_fragmentation,Alon-Streltsov-Cederbaum-PRA2008fragm,Lode-PRA2016}. Such an Ansatz is somewhat arbitrary, and in \cite{DimFalcOlg21-fragmented} it was objected that the error made by imposing an evolved fragmentation with constant number of orbitals at any time is in general not vanishing as $N\to\infty$.

 At \emph{finite gap} $\nu$, and in the mean-field limit in $N$, the emergence of \emph{infinite-rank} asymptotic marginals 
 \begin{equation}\label{eq:compactnesslimit}
  \gamma_{\infty,\nu,t}^{(k)}\;=\;\lim_{N\to\infty}\gamma_{N,\nu,t}^{(k)}
 \end{equation}
 in the course of the evolution of an initially fragmented state such as \eqref{eq:genfrag}, and the explicit form of $\gamma_{\infty,\nu,t}^{(k)}$ as a suitable average of rank-one projections $|\psi_t^{\otimes k}\rangle\langle\psi_t^{\otimes k}|$ for one-body orbitals $\psi_t$ all evolving according to the same Hartree (cubic-convolutive, non-linear Schr\"{o}dinger) equation 
 \begin{equation}
  \ii\partial_t\psi_t\;=\;\mathfrak{h}_\nu\psi_t+(V*|\psi_t|^2)\psi_t\,,
 \end{equation}
 was argued in \cite{DimFalcOlg21-fragmented} to follow \emph{directly} from the preceding analysis \cite{Ammari-Nier-2015_Wigner}, where it had been shown that in the mean-field quantum dynamics of a many-bosons system Wigner measures propagate along the nonlinear Hartree flow. In such works, \eqref{eq:compactnesslimit} is controlled by measure-theoretic compactness arguments, thereby inherently \emph{without} quantitative rate of convergence.

 Based on the latter result, it was further argued in \cite{DimFalcOlg21-fragmented} that in the \emph{infinite gap} limit $\nu\to\infty$, indifferently taken before or after the $N\to\infty$ limit, the one-body marginal $\gamma_{N,t}^{(1)}$ at any time $t>0$ of the mean-field evolution of \eqref{eq:genfrag} asymptotically attains  the same rank $\ell\geqslant 2$ it had at time zero, and with an explicit characterisation of the asymptotic matrix elements of $\gamma_{N,t}^{(1)}$ by means of an underlying Hartree dynamics, thus giving a first indication of persistence of finite-rank fragmented BEC self-consistently evolving according to the Hartree flow.

 \section{Many-body fragmented BEC and behaviour at the level of marginals}\label{sec:manybodyfrag}

  In view of the background outlined so far, we discuss a first result (Proposition \ref{prop:reduced-marginals-time-zero}, Corollary \ref{cor:rankofourstate}, and Proposition \ref{prop:asymptoticvicinity} below) that shows that monitoring a many-body bosonic state at the level of \emph{any} marginal is not sufficient to identify the presence of fragmentation.

  Henceforth, for clarity of presentation, we shall keep the \emph{two-level fragmentation} as a case study: straightforward generalisations to generic $\ell$-level fragmentation are possible.
  
 We consider the typical case where the one-body Hilbert space is
  \begin{equation}\label{eq:Honebody}
   \cH\;:=\;L^2(\mathbb{R}^3)
  \end{equation}
 and the $N$-body Hilbert space is therefore 
  \begin{equation}\label{eq:Hmanybody}
  \cH_N\;=\;L^2_{\mathrm{sym}}(\mathbb{R}^{3N})\;\equiv\;L^2_{\mathrm{sym}}(\mathbb{R}^{3N},\ud x_1,\dots,\ud x_N)
 \end{equation}
 (the space of Lebesgue square-integrable functions in the $N$ variables $x_1,\dots,x_N\in\mathbb{R}^3$ which are also symmetric under any permutation of variables).

 At the one-body level, we single out two states
  \begin{equation}\label{eq:f1f2}
  \varphi_1,\varphi_2\in L^2(\mathbb{R}^3)\qquad\textrm{such that}\qquad \|\varphi_1\|_{L^2}=\|\varphi_2\|_{L^2}=1\,,\qquad\langle\varphi_1,\varphi_2\rangle_{L^2}=0\,,
 \end{equation}
 namely the two one-body orbitals of possible fragmentation.

 Modelled on these data, we consider three $N$-body states -- a pure state $\Psi_N$ (with density matrix $\gamma_N:=|\Psi_N\rangle\langle\Psi_N|$) and two mixed states $\widetilde{\gamma}_N$ and $\rho_N$ -- defined, respectively, as
   \begin{equation}\label{eq:defPsiN}
  \Psi_N\;:=\;\varphi_{1}^{\otimes N_{1}}\vee\varphi_{2}^{\otimes N_{2}}\,,
 \end{equation}
  \begin{equation}\label{eq:ourauxiliary}
 \begin{split}
  \widetilde{\gamma}_{N}\;&:=\;\frac{1}{(2\pi)^{2}}  \iint_{[0,2\pi]^2}\ud\theta_1\,\ud\theta_2\:\big|(\psi_{\varphi_1,\varphi_2,N}^{\theta_1,\theta_2})^{\otimes N}\big\rangle\big\langle (\psi_{\varphi_1,\varphi_2,N}^{\theta_1,\theta_2})^{\otimes N}\big|\,, \\
  \psi_{\varphi_1,\varphi_2,N}^{\theta_1,\theta_2}\;&:=\;\sqrt{\frac{N_{1}}{N}}e^{-\mathrm{i}\theta_{1}}\varphi_{1}+\sqrt{\frac{N_{2}}{N}}e^{-\mathrm{i}\theta_{2}}\varphi_{2}\,,
  \end{split}
 \end{equation}
 and 
   \begin{equation}\label{eq:rho-auxiliary}
   \rho_N\;:=\;\frac{N_1}{N}\,|\varphi_1^{\otimes N}\rangle\langle\varphi_1^{\otimes N}|+\frac{N_2}{N}\,|\varphi_2^{\otimes N}\rangle\langle\varphi_2^{\otimes N}|\,,
  \end{equation}
 In \eqref{eq:defPsiN}-\eqref{eq:rho-auxiliary} above, $N_1,N_2\in\mathbb{N}$, $N_1+N_2=N$, and we assume a whole sequence of such states is built with $N_1\equiv N_1(N)$ and $N_2\equiv N_2(N)$ satisfying
 \begin{equation}
  \lim_{N\to\infty}\frac{N_1}{N}\:=:\:n_1\,\in\,(0,1)\,,\qquad \lim_{N\to\infty}\frac{N_2}{N}\:=:\:n_2\,\in\,(0,1)\,.
 \end{equation}
 Again, `$\vee$' denotes the overall symmetric tensor product between the two factors. 
 
  The $k$-body reduced density matrices associated with such states are controlled explicitly.

  \begin{prop}\label{prop:reduced-marginals-time-zero}
  For fixed $k\in\{1,\dots, N\}$ the $k$-marginals $\gamma_N^{(k)}$, $\widetilde{\gamma}_N^{(k)}$, and $\rho_N^{(k)}$ associated, respectively, to $\gamma_N$, $\widetilde{\gamma}_N$, and $\rho_N$ are given by
  \begin{eqnarray}
   \gamma_N^{(k)} &=& \sum_{j=0}^{k}\,c_{k,j}\,
    \big|\varphi_{1}^{\otimes(k-j)}\vee\varphi_{2}^{\otimes j}\big\rangle\big\langle\varphi_{1}^{\otimes(k-j)}\vee\varphi_{2}^{\otimes j}\big|\,, \qquad c_{k,j}\,:=\,\frac{\,\binom{N_1}{k-j} \binom{N_2}{j}\,}{\binom{N}{k}}\,,    \label{eq:gammak}\\
   \widetilde{\gamma}_N^{(k)}&=&\sum_{j=0}^{k}\,\widetilde{c}_{k,j}\,\big|\varphi_{1}^{\otimes(k-j)}\vee\varphi_{2}^{\otimes j}\big\rangle\big\langle\varphi_{1}^{\otimes(k-j)}\vee\varphi_{2}^{\otimes j}\big|\,,\qquad \widetilde{c}_{k,j}\,:=\,\binom{k}{j} \frac{N_{1}^{k-j}N_{2}^{j}}{N^{k}}\,, \label{eq:gammaktilde} \\
   \rho_N^{(k)}&=&\frac{N_1}{N}\,|\varphi_1^{\otimes k}\rangle\langle\varphi_1^{\otimes k}|+\frac{N_2}{N}\,|\varphi_2^{\otimes k}\rangle\langle\varphi_2^{\otimes k}|\,. \label{eq:rhok}
  \end{eqnarray}
  In particular, 
  \begin{equation}\label{eq:same1marginal}
   \gamma_N^{(1)}\;=\;\widetilde{\gamma}_N^{(1)}\;=\;\rho_N^{(1)}\;=\;\frac{N_1}{N}\,|\varphi_1\rangle\langle\varphi_1|+\frac{N_2}{N}\,|\varphi_2\rangle\langle\varphi_2|\,.
  \end{equation}
 \end{prop}

  \begin{cor}\label{cor:rankofourstate} One has 
  \begin{eqnarray}
   \mathrm{rank}(\gamma_N^{(k)}) &=& k+1 \,, \label{eq:rankorig}\\
   \mathrm{rank}(\widetilde{\gamma}_N^{(k)}) &=& k+1\,, \label{eq:rankauxiliary} \\
   \mathrm{rank}(\rho_N^{(k)})&=&2\,.
  \end{eqnarray}
 \end{cor}

  In addition, $\Psi_N$ and $\widetilde{\gamma}_N$ display asymptotic-in-$N$ closeness at the level of the reduce marginals. (Instead, $\rho_N$ does \emph{not} have the same feature.)

   \begin{prop}\label{prop:asymptoticvicinity}
   For each $k\in\{1,\dots,N\}$ there exists $c_k\geqslant 0$ such that
   \begin{equation}\label{eq:asymptoticvicinity}
    \mathrm{Tr}\big| \gamma_N^{(k)}- \widetilde{\gamma}_N^{(k)}\big|\;\leqslant\;\frac{c_k}{N}\,.
   \end{equation}
   In particular, $c_1=0$. Thus, at fixed $k\in\mathbb{N}$, $ \mathrm{Tr}| \gamma_N^{(k)}- \widetilde{\gamma}_N^{(k)}|=O(N^{-1})$ as $N\to\infty$.
  \end{prop}

 As argued in Sections \ref{sec:fragmentation-marginals}-\ref{sec:fragmented-manybody}, $\Psi_N$ given in \eqref{eq:defPsiN} is the prototype of a many-body state preparable with two-level fragmented BEC onto the two one-body orbitals $\varphi_1$ and $\varphi_2$, with (asymptotic) occupation numbers $n_1$ and $n_2$.

 Instead, neither $\widetilde{\gamma}_{N}$ nor $\rho_N$ correspond to an actual preparation of fragmented BEC, so none of them can be interpreted as a state of fragmentation, as opposite to $\Psi_N$. Both $\widetilde{\gamma}_{N}$ and $\rho_N$ are statistical superpositions of states of complete (100\%) simple BEC onto one-body orbitals that in the former case are the $\psi_{\varphi_1,\varphi_2,N}^{\theta_1,\theta_2}$'s, and in the latter are $\varphi_1$ and $\varphi_2$.

 This shows that the mere control of the occupation numbers at the level of the one-body marginals does not characterise fragmentation unambiguously.

 Besides, Corollary \ref{cor:rankofourstate} clarifies that even requiring (as in \cite{DimFalcOlg21-fragmented}) that the rank of all $k$-marginals is finite and not constant in $k$ (a constraint that clearly factors $\rho_N$ out), may still fail to select many-body states with an appropriate interpretation of fragmented BEC.
 
 We complete this Section with the proofs of the above statements.  
  Let us first discuss the proof of Propositions \ref{prop:reduced-marginals-time-zero} and \ref{prop:asymptoticvicinity}.
    
     \begin{proof}[Proof of Proposition \ref{prop:reduced-marginals-time-zero}]
    Formula \eqref{eq:gammak} is proved already in \cite[Proposition 2.6]{DimFalcOlg21-fragmented} when $\ell=2$. For completeness of presentation, we develop our own proof by demonstrating Proposition \ref{prop:spinor-k-marginal} below: indeed, as argued therein, Propositions \ref{prop:reduced-marginals-time-zero} and \ref{prop:spinor-k-marginal} are in fact the same statements, up to Hilbert space isomorphism.    
    For $\rho_N^{(k)}$, the proof of \eqref{eq:rhok} is straightforward from \eqref{eq:tracing-out} and \eqref{eq:rho-auxiliary}. Concerning \eqref{eq:gammaktilde}, we apply again \eqref{eq:tracing-out}, now to \eqref{eq:ourauxiliary}, and find, for $k\in\{1,\dots,N-1\}$,
        \begin{align*}
    	& \widetilde{\gamma}_{N}^{(k)}(x_{1},\dots,x_{k};y_{1},\dots,y_{k})\\
    	&=\;\frac{1}{(2\pi)^{2}}\iint_{[0,2\pi]^{2}}\mathrm{d}\theta_{1}\mathrm{d}\theta_{2}\,\prod_{j=1}^k \left(\sqrt{\frac{N_{1}}{N}}e^{-\mathrm{i}\theta_{1}}\varphi_{1}(x_j)+\sqrt{\frac{N_{2}}{N}}e^{-\mathrm{i}\theta_{2}}\varphi_{2}(x_j)\right)\left(\sqrt{\frac{N_{1}}{N}}e^{\mathrm{i}\theta_{1}}\overline{\varphi_{1}(y_j)}+\sqrt{\frac{N_{2}}{N}}e^{\mathrm{i}\theta_{2}}\overline{\varphi_{2}(y_j)}\right)\\ 
    	&\qquad\qquad\qquad\qquad \times \int_{\mathbb{R}^3}\cdots\int_{\mathbb{R}^3}\mathrm{d}x_{k+1}\dots\mathrm{d}x_{N}\prod_{\ell=k+1}^N \left|\psi_{\varphi_1,\varphi_2,N}^{\theta_1,\theta_2}(x_\ell) \right|^2\,.
    \end{align*}
      The integration in the spatial variables $x_{k+1},\dots,x_N$ cancels out because, owing to \eqref{eq:f1f2},
      \[
       \big\| \psi_{\varphi_1,\varphi_2,N}^{\theta_1,\theta_2} \big\|_{\cH}^2\;=\;\frac{N_1}{N}\|\varphi_1\|_{L^2}^2+\frac{N_2}{N}\|\varphi_2\|_{L^2}^2\;=\;1\,.
      \]
      Thus, with the shorthand $\phi_p:=\sqrt{\frac{N_p}{N}}\,e^{-\ii\theta_p}\varphi_p$, $p\in\{1,2\}$,
      \[
       \widetilde{\gamma}_{N}^{(k)}(x_{1},\dots,x_{k};y_{1},\dots,y_{k})\;=\;\frac{1}{(2\pi)^{2}}\iint_{[0,2\pi]^{2}}\mathrm{d}\theta_{1}\mathrm{d}\theta_{2}\,\prod_{j=1}^k\big(\phi_1(x_j)+\phi_2(x_j)\big)\big(\overline{\phi_1(y_j)}+\overline{\phi_2(y_j)}\big)\,.
      \]
      We re-write
      \[
       \prod_{j=1}^k\big(\phi_1(x_j)+\phi_2(x_j)\big)\;=\;\sum_{\sigma\in\{1,2\}^k}\prod_{\ell=1}^k\phi_{\sigma_\ell}(x_\ell)
      \]
      with the notation $\sigma\equiv(\sigma_1,\dots,\sigma_k)$ for an element of $\{1,2\}^k$, i.e., an ordered collection of $1$'s and $2$'s. Therefore,
      \[
       \begin{split}
        &\widetilde{\gamma}_{N}^{(k)}(x_{1},\dots,x_{k};y_{1},\dots,y_{k})\;=\;\frac{1}{(2\pi)^{2}}\iint_{[0,2\pi]^{2}}\mathrm{d}\theta_{1}\mathrm{d}\theta_{2}\,\bigg(\sum_{\sigma\in\{1,2\}^k}\prod_{\ell=1}^k\phi_{\sigma_\ell}(x_\ell)\bigg)\bigg(\sum_{\sigma'\in\{1,2\}^k}\prod_{m=1}^k\overline{\phi_{\sigma'_m}(y_m)}\bigg) \\
        &\qquad =\;\frac{1}{(2\pi)^{2}}\sum_{ \sigma,\sigma'\in\{1,2\}^k }N^{-k}\prod_{\ell,m=1}^k\sqrt{ N_{\sigma_\ell} N_{\sigma'_m}\,}\;\varphi_{\sigma_\ell}(x_\ell)\,\overline{\varphi_{\sigma'_m}(y_m)}\iint_{[0,2\pi]^{2}}\mathrm{d}\theta_{1}\mathrm{d}\theta_{2}\,e^{\ii\big(\sum_{\ell=1}^k\theta_{\sigma_\ell}-\sum_{m=1}^k\theta_{\sigma'_m}\big)}\,.
       \end{split}
      \]
     If $\sigma$ and $\sigma'$ consist of a \emph{different} number of $1$'s (and hence of $2$'s), the exponential in the integral above takes the form $e^{\ii q (\theta_1-\theta_2)}$ for some $q\in\mathbb{Z}\setminus\{0\}$, and therefore the integration over $\theta_1$ gives zero. Only those $\sigma,\sigma'$ with $|\sigma^{-1}(1)|=|{\sigma'}^{-1}(1)|$ (i.e., the same number of $1$'s in the two collections) contribute to the above expression, and for each such pair the exponential in the integral trivialises to 1 and the double integration in $\theta_1,\theta_2$ cancels out by means of the normalisation factor $(2\pi)^{-2}$. Observe also that for any two $\sigma,\sigma'$ with $|\sigma^{-1}(1)|=|{\sigma'}^{-1}(1)|$ one has
     \[
      \prod_{\ell,m=1}^k \sqrt{N_{\sigma_\ell} N_{\sigma'_m}\,}\;=\; \prod_{\ell=1}^k N_{\sigma_\ell}\,.
     \]
    The preceding considerations yield the first line of the following chain of identities:
    \[
     \begin{split}
      & \widetilde{\gamma}_{N}^{(k)}(x_{1},\dots,x_{k};y_{1},\dots,y_{k})\;=\;\!\!\!\sum_{\substack{ \sigma,\sigma'\in\{1,2\}^k \\ |\sigma^{-1}(1)|=|{\sigma'}^{-1}(1)|}}\!\!\!\frac{N_{\sigma_1}\cdots N_{\sigma_k}}{N^k}\;  \prod_{\ell,m=1}^k \varphi_{\sigma_\ell}(x_\ell)\,\overline{\varphi_{\sigma'_m}(y_m)} \\
      &\qquad\qquad =\;\sum_{j=0}^k\sum_{\substack{ \sigma,\sigma'\in\{1,2\}^k \\ |\sigma^{-1}(1)|=|{\sigma'}^{-1}(1)| = k-j}}\!\!\!\frac{N_1^{k-j} N_2^j}{N^k}\;   \prod_{\ell,m=1}^k \varphi_{\sigma_\ell}(x_\ell)\,\overline{\varphi_{\sigma'_m}(y_m)} \\
       &\qquad\qquad =\; \sum_{j=0}^k\:\frac{N_1^{k-j} N_2^j}{N^k} \left( \sum_{\substack{ \sigma\in\{1,2\}^k \\ |\sigma^{-1}(1)|= k-j}}\prod_{\ell=1}^k \varphi_{\sigma_\ell}(x_\ell) \right)\left( \sum_{\substack{ \sigma'\in\{1,2\}^k \\ |{\sigma'}^{-1}(1)|= k-j}}\prod_{m=1}^k \overline{\varphi_{\sigma'_m}(y_m)} \right).
     \end{split}
    \]
    On the other hand, the vector $\varphi_1^{\otimes(k-j)}\vee\varphi_2^{\otimes j}$ consists, by definition, of the normalised sum of all functions obtained from $\varphi_1(x_1)\cdots\varphi_{1}(x_{k-j})\varphi_2(x_{k-j+1})\cdots\varphi_2(x_k)$ by applying all possible permutations with repetitions of the index set $\{1,\dots,1,2,\dots,2\}$ with $k-j$ $1$'s and $j$ $2$'s. There are precisely $\binom{k}{j}$ such distinct functions, all of the form $\prod_{\ell=1}^k \varphi_{\sigma_\ell}(x_\ell)$ for some $\sigma\in\{1,2\}^k$ with $|{\sigma}^{-1}(1)|= k-j$, and they are all pair-wise orthonormal (owing to \eqref{eq:f1f2}), implying that the normalisation factor for the sum thus obtained is $\binom{k}{j}^{-\frac{1}{2}}$. Thus,
    \[
     \big(\varphi_1^{\otimes(k-j)}\vee\varphi_2^{\otimes j}\big)(x_1,\dots,x_k)\;=\;\binom{k}{j}^{-\frac{1}{2}}\!\!\!\!\sum_{\substack{ \sigma\in\{1,2\}^k \\ |\sigma^{-1}(1)|= k-j}}\prod_{\ell=1}^k \varphi_{\sigma_\ell}(x_\ell)
    \]
    and an analogous formula holds for $\big(\varphi_1^{\otimes(k-j)}\vee\varphi_2^{\otimes j}\big)(y_1,\dots,y_k)$. This allows to re-write
    \[
     \widetilde{\gamma}_{N}^{(k)}(x_{1},\dots,x_{k};y_{1},\dots,y_{k})\;=\;\sum_{j=0}^k\:\binom{k}{j}\frac{N_1^{k-j} N_2^j}{N^k}\:\big(\varphi_1^{\otimes(k-j)}\vee\varphi_2^{\otimes j}\big)(x_1,\dots,x_k)\,\overline{\big(\varphi_1^{\otimes(k-j)}\vee\varphi_2^{\otimes j}\big)(y_1,\dots,y_k)}\,,
    \]
   that is,
   \[
    \widetilde{\gamma}_{N}^{(k)}\;=\;\sum_{j=0}^{k}
    \binom{k}{j}\frac{N_{1}^{k-j}N_{2}^{j}}{N^{k}} \;\big|\varphi_{1}^{\otimes(k-j)}\vee\varphi_{2}^{\otimes j}\big\rangle\big\langle\varphi_{1}^{\otimes(k-j)}\vee\varphi_{2}^{\otimes j}\big|\,.
   \]
   Formula \eqref{eq:gammaktilde} is finally established.
    \end{proof}

   In order to prove Proposition \ref{prop:asymptoticvicinity} it is convenient to single out the following estimate.
   
    \begin{lem}\label{lem:combinatorial}
    For each $N\in \mathbb{N}$, $N\geqslant 2$, let $N_1,N_2\in\mathbb{N}$ such that $N_1+N_2=N$ and $\kappa N\leqslant N_1< N$, $\kappa\leqslant N_2< N$ for some $\kappa\in(0,1)$. For each $k\in\{0,1,\dots, N-1\}$ there exist $a_k>0$, depending only on $k$ (and $\kappa$), such that, for each $j\in\{1,\dots,k\}$,
    \begin{equation}\label{eq:combinatorial}
    \left| \binom{N}{k}^{-1} \binom{N_1}{k-j}
    \binom{N_2}{k} - \binom{k}{j} \frac{N_1^{k-j}N_2^j}{N^k}\right| \;\leqslant\; \frac{a_k}{N}\,.
    \end{equation}
    \end{lem}

    \begin{proof}
    We compute
    \begin{align*}
    \binom{N}{k}^{-1}
    \binom{N_1}{k-j}\binom{N_2}{j}\;&=\;\frac{k!(N-k)!}{N!}\cdot\frac{N_{1}!}{(k-j)!(N_{1}-k+j)!}\cdot\frac{N_{2}!}{j!(N_{2}-j)!}\\
    & =\;\binom{k}{j}
    \frac{(N-k)!N_{1}!N_{2}!}{N!(N_{1}-k+j)!(N_{2}-j)!}\\
    & =\;\binom{k}{j}\frac{N_{1}(N_{1}-1)\cdots(N_{1}-k+j+1)N_{2}(N_{2}-1)\cdots(N_{2}-j+1)}{N(N-1)\cdots(N-k+1)}\,..
    \end{align*}
    Then, by Taylor expansion for large $N$,
    \begin{align*}
     \frac{N_{1}^{k-j}N_{2}^{j}}{N^{k}}&-\frac{N_{1}(N_{1}-1)\cdots(N_{1}-k+j+1)N_{2}(N_{2}-1)\cdots(N_{2}-j+1)}{N(N-1)\cdots(N-k+1)}\\
     & =\;\frac{N_{1}^{k-j}N_{2}^{j}}{N^{k}}\left(1-\frac{(1-\frac{1}{N_{1}})\cdots(1-\frac{k-j-1}{N_{1}})(1-\frac{1}{N_{2}})\cdots(1-\frac{j-1}{N_{2}})}{(1-\frac{1}{N})\cdots(1-\frac{k-1}{N})}\right)\\
     & =\;\frac{N_{1}^{k-j}N_{2}^{j}}{N^{k}}\,\Big(\frac{1}{N_{1}}+\dots+\frac{k-j-1}{N_{1}}+\frac{1}{N_{2}}+\dots+\frac{j-1}{N_{2}}\\
     &\qquad\qquad\qquad\quad +\frac{1}{N}+\dots+\frac{k-1}{N}
     +O\Big(\frac{1}{N_1 N_2}+ \frac{1}{N_1 N} + \frac{1}{N_2 N}\Big)
     \Big)\; \leqslant\; \frac{b_k}{N}
    \end{align*}
    for some $b_k>0$ only depending on $k$ (and $\kappa$). The final result then follows by setting $a_k:=b_k\max_j\binom{k}{j}$.
    \end{proof}
    
    \begin{proof}[Proof of Proposition \ref{prop:asymptoticvicinity}]
    On account of \eqref{eq:gammak}-\eqref{eq:gammaktilde},
    \[
     \begin{split}
      \gamma_{N}^{(k)} \;&=\;\sum_{j=0}^{k}
    \binom{N}{k}^{-1} \binom{N_1}{k-j} \binom{N_2}{j} 
    \big|\varphi_{1}^{\otimes(k-j)}\vee \varphi_{2}^{\otimes j}\big\rangle \big\langle \varphi_{1}^{\otimes(k-j)}\vee\varphi_{2}^{\otimes j}\big|\,, \\
     \widetilde{\gamma}_{N}^{(k)}\;&=\;\sum_{j=0}^{k}
    \binom{k}{j}\frac{N_{1}^{k-j}N_{2}^{j}}{N^{k}} \;\big|\varphi_{1}^{\otimes(k-j)}\vee\varphi_{2}^{\otimes j}\big\rangle\big\langle\varphi_{1}^{\otimes(k-j)}\vee\varphi_{2}^{\otimes j}\big|\,.
     \end{split}
    \]
    Then the difference $\gamma_{N}^{(k)} - \widetilde{\gamma}_{N}^{(k)}$ is given by $k+1$ summands, each consisting of a rank-one orthogonal projection multiplied by a coefficient estimated precisely by \eqref{eq:combinatorial} of Lemma \ref{lem:combinatorial}. This yields
    \[
    \Tr \Big| \gamma_{N}^{(k)} - \widetilde{\gamma}_{N}^{(k)} \Big|
    \;\leqslant\;
    \sum_{j=0}^{k} \frac{a_{k}}{N} \;\leqslant\; \frac{c_{k}}{N}\,,
    \]
    having set $c_k := (k+1) a_k$.
    \end{proof}

    \begin{rem}\label{rem:general-ell}
     Although for the present discussion we picked non-restrictively the case of $\ell=2$ levels of fragmentation, we can prove the following $\ell$-level generalisations, with reasonings that are analogous to the preceding ones. Similarly to Proposition \ref{prop:reduced-marginals-time-zero}, and for $\ell\in\mathbb{N}$, $\ell\geqslant 2$, introduce the purely $\ell$-level fragmented state 
    \begin{equation}\label{eq:ell1}
     \Psi_N\;:=\;\varphi_1^{\otimes N_1}\vee\cdots\vee\varphi_{\ell}^{\otimes N_{\ell}}
    \end{equation}
    (see \eqref{eq:genfrag} above) with one-body states of fragmentation given by the orthonormal system $\{\varphi_1,\dots,\varphi_{\ell}\}$ in $\cH$, and populations $N_1\equiv N_1(N)$, $\dots,$ $N_\ell\equiv N_{\ell}(N)$ such that $N_1+\cdots+N_\ell=N$ and $N_j/N\xrightarrow{N\to\infty} n_j\in(0,1)$, $j\in\{1,\dots,\ell\}$, as well as the non-fragmented $N$-body mixture
    \begin{equation}\label{eq:ell2}
 \begin{split}
  \widetilde{\gamma}_{N}\;&:=\frac{1}{(2\pi)^{\ell}}  \int\cdots\int_{[0,2\pi]^\ell}\ud\theta_1\cdots\ud\theta_{\ell}\:\big|(\psi_{\varphi_1,\cdots,\varphi_\ell,N}^{\theta_1,\cdots,\theta_\ell})^{\otimes N}\big\rangle\big\langle (\psi_{\varphi_1,\dots,\varphi_\ell,N}^{\theta_1,\dots,\theta_\ell})^{\otimes N}\big|\,, \\
  \psi_{\varphi_1,\dots,\varphi_\ell,N}^{\theta_1,\dots,\theta_\ell}\;&:=\;\sum_{j-=1}^\ell\sqrt{\frac{N_{j}}{N}}e^{-\mathrm{i}\theta_{j}}\varphi_{j}\,.
  \end{split}
 \end{equation}
     For $k\in\{1,\dots,N-1\}$ introduce also the index sets
     \begin{equation}
      \mathcal{F}_k^N\;:=\;\Big\{\mathbf{a}\equiv(a_1,\dots,a_\ell)\,\Big|\,a_j\in\{1,\dots,N_j\}\:\textrm{ for each }\:j\in\{1,\dots,\ell\},\:\textrm{ and }\:\sum_{j=1}^\ell a_j=k\Big\}\,.
     \end{equation}
  Then the $k$-marginals associated with the $N$-body states \eqref{eq:ell1} and \eqref{eq:ell2} above are, respectively,
  \begin{equation}
   \begin{split}
    \gamma_N^{(k)}\;&=\;\sum_{\mathbf{a}\in\mathcal{F}_k^N} c_{\mathbf{a}}\,\big|\varphi_1^{\otimes a_1}\vee\cdots\vee\varphi_{\ell}^{\otimes a_\ell}\big\rangle\big\langle \varphi_1^{\otimes a_1}\vee\cdots\vee\varphi_{\ell}^{\otimes a_\ell}\big|\,,\qquad c_{\mathbf{a}}\,:=\,\binom{N}{k}^{-1}\binom{N_1}{a_1}\cdots\binom{N_\ell}{a_\ell}
   \end{split}
  \end{equation}
 and
 \begin{equation}
  \widetilde{\gamma}_N^{(k)}\;=\;\sum_{\mathbf{a}\in\mathcal{F}_k^N} \widetilde{c}_{\mathbf{a}}\,\big|\varphi_1^{\otimes a_1}\vee\cdots\vee\varphi_{\ell}^{\otimes a_\ell}\big\rangle\big\langle \varphi_1^{\otimes a_1}\vee\cdots\vee\varphi_{\ell}^{\otimes a_\ell}\big|\,,\qquad \widetilde{c}_{\mathbf{a}}\,:=\,\frac{k!}{\,a_{1}!\cdots a_{\ell}!\,} \frac{N_{1}^{a_{1}}\cdots N_{\ell}^{a_{\ell}}}{N}\,. 
 \end{equation}
  Proposition \ref{prop:asymptoticvicinity} and Lemma \ref{lem:combinatorial} can be then reproduced with identical formulation, that is,
  \begin{equation}
   \big| c_{\mathbf{a}} - \widetilde{c}_{\mathbf{a}} \big|\;\leqslant\;\frac{\mathsf{a}_k}{N}\,,
  \end{equation}
  and consequently the above $k$-marginals are as close as 
     \begin{equation}
    \mathrm{Tr}\big| \gamma_N^{(k)}- \widetilde{\gamma}_N^{(k)}\big|\;\leqslant\;\frac{c_k}{N}\,,
   \end{equation}
   for suitable constants $\mathsf{a}_k$ and $c_k$. 
   \end{rem}

 \section{Quantitative emergence of effective dynamics of fragmentation at infinite  gap}\label{sec:scenarioINfinitegap}

  The second type of result of the present work concerns the time evolution of a many-body state with initial fragmented BEC, when the dynamics is monitored at the level of marginals, in the infinite particle limit \emph{and} the infinite gap limit together.
  
  The physically meaningful setting would be the one of the previous Section, with initial state $\Psi_N\in\cH$ given by \eqref{eq:defPsiN} (exact two-level fragmentation) and evolution governed by a many-body Hamiltonian of the form
  \begin{equation}
  \sum_{j=1}^N \big(-\Delta_{x_j}+U_\nu(x_j)\big)\,+\frac{1}{N}\!\!\sum_{1\leqslant \ell<r\leqslant N }\!\!V(x_\ell-x_r)
 \end{equation}
 (acting on $\cH_N$), where $V$ is an inter-particle interaction potential and $U_\nu$ is a trapping potential such that the occupied orbitals $\varphi_1$ and $\varphi_2$ are, respectively, the ground state and the first excited state of the one-body Hamiltonian $-\Delta_{x_j}+U_\nu$, and are separated by a gap $\nu>0$, and moreover a typical mean-field scaling factor $1/N$ is inserted. Denoting by $\Psi_{N,\nu,t}$ the many-body state at later times $t>0$ and by $\gamma_{N,\nu,t}^{(k)}$ its $k$-marginal, one would like to characterise the object $\lim_{N,\nu\to\infty}\gamma_{N,\nu,t}^{(k)}$.

 There are multiple factors making such an analysis complicated and too general. We therefore perform it for  
  a tractable \emph{toy model}, specifically the toy model recently introduced in \cite{DimFalcOlg21-fragmented} for precisely the same type of question. However, in \cite{DimFalcOlg21-fragmented} the overall conclusion is non-quantitative, as the $N\to\infty$ limit is only controlled by measure-theoretic compactness arguments, hence without rate of convergence. Here we reproduce the result through an alternative route and with rates of convergence. In fact, the toy model below allows for a caricature of the infinite gap limit in the evolution of a fragmented condensate which is not completely satisfactory from the physical viewpoint, and yet it retains an amount of instructiveness and insight.

    For the present purposes one-body and $N$-body Hilbert spaces are modified from \eqref{eq:Honebody}-\eqref{eq:Hmanybody} to
  \begin{equation}\label{eq:nuHilbertspaces}
    \cH\;:=\;L^2(\mathbb{R}^3)\otimes\mathbb{C}^2\,,\qquad \cH_N\;:=\;\cH^{\otimes_{\mathrm{sym}} N}\,.
   \end{equation}
  Here the idea is to model fragmentation occurring on one-body states that differ for their \emph{spinor} component: the factor $\mathbb{C}^2$ is required for the two-level fragmentation discussed here, for generic $\ell$-level fragmentation we should use $\mathbb{C}^{\ell}$.

  We also consider the one-body Hamiltonian
  \begin{equation}\label{eq:hnutensorized}
   \mathfrak{h}_\nu\;:=\;{\textstyle\frac{1}{2}}\nu(-\Delta+x^2-3)\otimes\mathbbm{1}\,,\qquad \nu>0\,,
  \end{equation}
  self-adjointly realised in $\cH$ with domain $H^{1,1}(\mathbb{R}^3)\otimes\mathbb{C}^2$, where 
  \begin{equation}
   H^{1,1}(\mathbb{R}^3)\;:=\;\big\{\psi\in L^2(\mathbb{R}^3)\,|\,\|\nabla\psi\|_{L^2}^2+\|x\psi\|_{L^2}^2<+\infty\big\}\,.
  \end{equation}
  Correspondingly, the $N$-body mean field Hamiltonian shall be
  \begin{equation}\label{eq:hamiltonianNnunu}
  H_{N,\nu}\;:=\;\sum_{j=1}^N \mathfrak{h}_{\nu,j}\,+\frac{1}{N}\!\!\sum_{1\leqslant \ell<r\leqslant N }\!\!V(x_\ell-x_r)\,,
 \end{equation}
  with assumptions on $V$ that ensure the self-adjointness of $H_{N,\nu}$ in $\cH_N$, where $\mathfrak{h}_{\nu,j}$ acting as $\mathfrak{h}_\nu$ on the $j$-th particle and as the identity on all others. 
  Letting $\mathfrak{h}_\nu$ to be trivial on the `spin' sector is a stratagem used in \cite{DimFalcOlg21-fragmented} to introduce an obvious spectral degeneracy, including the therefore doubly degenerate ground state. The particular form \eqref{eq:hnutensorized} of $h_\nu$ satisfies the \emph{gap condition}
  \begin{equation}\label{eq:gapcondition}
  \inf\sigma(\mathfrak{h}_\nu)\,=\,0\,,\qquad \inf(\sigma(\mathfrak{h}_\nu)\setminus\{0\})\,=\,\nu
 \end{equation}
 (see \eqref{eq:gapfirst} above), as well as the fact that the domain of $\mathfrak{h}_\nu$ is \emph{independent} of $\nu$. Moreover, it has the technically relevant feature that the two-dimensional ground state subspace of $\mathfrak{h}_\nu$ is independent of $\nu$, being spanned by the orthonormal basis
  \begin{equation}\label{eq:spanningbasis}
   \varphi_1\;:=\;(2\pi)^{-\frac{3}{2}}\,e^{-x^2/2}\begin{pmatrix} 1 \\ 0 \end{pmatrix}\,,\qquad \varphi_2\;:=\;(2\pi)^{-\frac{3}{2}}\,e^{-x^2/2}\begin{pmatrix} 0 \\ 1 \end{pmatrix}\,.
  \end{equation}

  In terms of $\varphi_1,\varphi_2$ above, and for $N_1,N_2\in\mathbb{N}$ with $N_1+N_2=N$, we consider the $N$-body state
  \begin{equation}\label{eq:defPsiN2}
  \Psi_N\;:=\;\varphi_{1}^{\otimes N_{1}}\vee\varphi_{2}^{\otimes N_{2}}\,\in\,\cH_N\,.
 \end{equation}
  The goal now is to monitor the Schr\"{o}dinger evolution of $\Psi_N$ governed by $H_{N,\nu}$, at the level of marginals, identifying the leading dynamics, up to a sub-leading correction when $N$ and $\nu$ are large.

  We should warn the reader that the model \eqref{eq:nuHilbertspaces}-\eqref{eq:defPsiN2} is only to be regarded as informative (as we shall see in a moment) on the general mechanism of dynamical emergence of finite-rank marginals in the infinite gap limit; other than that, the physical meaningfulness of the scaling in $\nu$ in the Hamiltonian \eqref{eq:hnutensorized} is questionable, because while letting the gap $\nu$ to infinity one is also re-scaling the mass with a vanishing factor $\nu^{-1}$.

  The final result, that will be proved in Sections \ref{sec:proofofThmdynamics} and \ref{sec:effectiveMFdynamics}, is the following.
  
  \begin{thm}\label{thm:N-nu-dynamics}
   Let $\nu_0>0$. For each $t\geqslant 0$ let
  \begin{equation}
   \Psi_{N,\nu,t}\,:=\,e^{-\ii t H_{N,\nu}}\Psi_N\,,
  \end{equation}
 with $H_{N,\nu}$ defined in \eqref{eq:hnutensorized}-\eqref{eq:hamiltonianNnunu} and $\Psi_N$ defined in \eqref{eq:spanningbasis}-\eqref{eq:defPsiN2} for $N\in\mathbb{N}$, $N\geqslant 2$, $\nu\geqslant\nu_0$, and for two sequences $N_1\equiv N_1(N)$, $N_2\equiv N_2(N)$ in $\mathbb{N}$ such that $N_1+N_2=N$ and 
 \begin{equation}
  \Big|\frac{N_1}{N}-n_1\Big|\,=\,O(N^{-1})\,=\,\Big|\frac{N_2}{N}-n_2\Big|\qquad\textrm{as }\;N\to\infty
 \end{equation}
  for given $n_1,n_2\in(0,1)$. In \eqref{eq:hamiltonianNnunu} it is assumed that $V:\mathbb{R}^3\to\mathbb{R}$ is a measurable function such that $V(-x)=V(x)$ for a.e.~$x$ and
   \begin{equation}\label{eq:assumptionV2}
    V^2\;\lesssim\;-\Delta+x^2+\mathbbm{1} 
   \end{equation}
   in the sense of quadratic forms on $L^2(\mathbb{R}^3)$, on smooth and compactly supported functions of $\mathbb{R}^3$. Then $H_{N,\nu}$ is essentially self-adjoint and lower semi-bounded on the domain of smooth and compactly supported functions of $\cH_N$, and for each $\theta_1,\theta_2\in[0,2\pi]$ the initial value problem
   \begin{equation}\label{eq:IVPhartree-nu}
    \begin{split}
     \ii\partial_t \varphi^{\theta_1,\theta_2}_{\nu} \,&=\,\mathfrak{h}_{\nu}\,\varphi^{\theta_1,\theta_2}_{\nu} +\big(V*|\varphi^{\theta_1,\theta_2}_{\nu} |^2\big)\varphi^{\theta_1,\theta_2}_{\nu} \,, \\
     &\qquad \varphi^{\theta_1,\theta_2}_{\nu} \equiv\varphi^{\theta_1,\theta_2}_{\nu} (t,x),\quad (t,x)\in[0,+\infty)\times\mathbb{R}^3\,, \\
     \varphi^{\theta_1,\theta_2}_{\nu} (0,\cdot)\,&=\,\sqrt{n_1}\,e^{-\ii\theta_1}\,\varphi_1+\sqrt{n_2}\,e^{-\ii\theta_2}\varphi_2
    \end{split}
   \end{equation}
   (the convolution above being meant in the $x$-variable) is well posed in $C([0,+\infty),\mathcal{D}[\mathfrak{h}_\nu])$\,, in the sense that there exists a unique solution
   \begin{equation}\label{eq:phinuspace}
    \varphi^{\theta_1,\theta_2}_{\nu}\,\in\, C([0,+\infty),\mathcal{D}[\mathfrak{h}_\nu])\cap C^1([0,+\infty),\mathcal{D}[\mathfrak{h}_\nu]^*)
   \end{equation}
   to \eqref{eq:IVPhartree-nu}, with continuous dependence on initial data.
    Correspondingly, for each $k\in\{1,\dots,N-1\}$ let $\gamma_{N,\nu,t}^{(k)}$ be the $k$-body reduced density matrix associated with $\Psi_{N,\nu,t}$. Then there exist constants $A_k>0$, depending only on $k$ (and on $\nu_0$) and independent of $N\in\mathbb{N}$ with $N\geqslant 2$, of $\nu\geqslant\nu_0$, and of $t\geqslant 0$, such that, for every such $N,\nu,t$,
  \begin{equation}\label{eq:traceconvergence-nu}
   \mathrm{Tr}\,\big| \,\gamma_{N,\nu,t}^{(k)} - \gamma_{\infty,\nu,t}^{(k)}\,\big|\;\leqslant\;\frac{A_k\,e^{A_k t}}{N}\,,
  \end{equation}
  where 
  \begin{equation}\label{eq:gammainfk-nu}
   \gamma_{\infty,\nu,t}^{(k)}\,:=\,\frac{1}{\:(2\pi)^2}\iint_{[0,2\pi]^2}\ud\theta_1\,\ud\theta_2 \big|\varphi^{\theta_1,\theta_2}_{\nu}(t,\cdot)^{\otimes k}\big\rangle\big\langle \varphi^{\theta_1,\theta_2}_{\nu}(t,\cdot)^{\otimes k}\big|\,.
  \end{equation}
 The trace in \eqref{eq:traceconvergence-nu} and the ket-bra notation in \eqref{eq:gammainfk-nu} are understood in the Hilbert space $\cH_k=(L^2(\mathbb{R}^3)\otimes\mathbb{C}^2)^{\otimes_{\mathrm{sym}}k}$.
  \end{thm}

  Theorem \ref{thm:N-nu-dynamics} is designed so as to be combined with the following result from \cite{DimFalcOlg21-fragmented}.

  \begin{thm}[\cite{DimFalcOlg21-fragmented}]\label{thm:dimonte}
   For $\nu>0$ let $\mathfrak{h}_\nu$ be the operator \eqref{eq:hnutensorized} in $L^2(\mathbb{R}^3)\otimes \mathbb{C}^2$ with self-adjointness domain $H^{1,1}(\mathbb{R}^3)\otimes \mathbb{C}^2$, and for each $t\geqslant 0$ consider $\gamma_{\infty,\nu,t}^{(1)}$ as defined in \eqref{eq:gammainfk-nu} through \eqref{eq:IVPhartree-nu}-\eqref{eq:phinuspace} for $\varphi_1,\varphi_2$ from \eqref{eq:spanningbasis} and for given even measurable function $V:\mathbb{R}^3\to\mathbb{R}$ satisfying \eqref{eq:assumptionV2}. Then there exists a constant $A>0$, independent of $\nu$ and $t$, such that
   \begin{equation}\label{eq:traceconvergence-infnu}
     \mathrm{Tr}\,\big| \,\gamma_{\infty,\nu,t}^{(1)} - \gamma_{\infty,\infty,t}^{(1)}\,\big|\;\leqslant\;\frac{\,A\,e^{A t}}{\sqrt{\nu}\,}\,,
   \end{equation}
  where 
  \begin{equation}\label{eq:infgapdynamics1}
  \begin{split}
   \gamma_{\infty,\infty,t}^{(1)}\;&:=\;\sum_{j,\ell=1}^2 K_{j,\ell}(t) |\varphi_j\rangle\langle\varphi_\ell|\,, \\
   K_{j,\ell}(t)\;&:=\;\frac{1}{\:(2\pi)^2}\iint_{[0,2\pi]^2}\ud\theta_1\,\ud\theta_2\,\overline{\kappa^{\theta_1,\theta_2}_\ell(t)}\,\kappa^{\theta_1,\theta_2}_j(t)\,,
  \end{split}
  \end{equation}
 and where $\kappa_1^{\theta_1,\theta_2},\kappa_2^{\theta_1,\theta_2}$ solve the system of ordinary differential equations
 \begin{equation}\label{eq:infgapdynamics2}
  \begin{split}
   \ii\partial_t \kappa_1^{\theta_1,\theta_2}\,&=\,\big\langle \varphi_1,\big(V*|\Phi^{\theta_1,\theta_2}|^2\big)\Phi^{\theta_1,\theta_2}\big\rangle_{L^2(\mathbb{R}^3)\otimes\mathbb{C}^2}\,, \\
   \ii\partial_t \kappa_2^{\theta_1,\theta_2}\,&=\,\big\langle \varphi_2,\big(V*|\Phi^{\theta_1,\theta_2}|^2\big)\Phi^{\theta_1,\theta_2}\big\rangle_{L^2(\mathbb{R}^3)\otimes\mathbb{C}^2}\,, \\
   \Phi^{\theta_1,\theta_2}\,&:=\,\kappa_1^{\theta_1,\theta_2}\varphi_1+\kappa_2^{\theta_1,\theta_2}\varphi_2\,, \\   
   \kappa_1^{\theta_1,\theta_2}(0)\,&=\,\sqrt{n_1}\,e^{-\ii\theta_1}\,,\quad \kappa_2^{\theta_1,\theta_2}(0)\,=\,\sqrt{n_2}\,e^{-\ii\theta_2}\,.
  \end{split}
 \end{equation}
  Moreover,
  \begin{equation}
   \mathrm{rank}(\gamma_{\infty,\infty,t}^{(1)})\;=\;2\qquad \textrm{for a.e.~}t\in\mathbb{R}\,.
  \end{equation}
  \end{thm}

  \begin{rem}
   At the basis of the validity of the analysis of \cite{DimFalcOlg21-fragmented} yielding Theorem \ref{thm:dimonte}, an analysis that covers other variants of the prototype toy model \eqref{eq:nuHilbertspaces}-\eqref{eq:defPsiN2}, are the following features, that are indeed matched by the $\nu$-gapped one-body Hamiltonian $\mathfrak{h}_\nu$: $\mathfrak{h}_\nu$ has non-trivial action on the spatial variables only; its domain is $\nu$-independent; its ground state energy is zero irrespectively of the gap $\nu$; the ground state eigenspace is independent of $\nu$; the potential $V$ is Kato small with respect to $\mathfrak{h}_\nu$.   
  \end{rem}

   An obvious triangular inequality between \eqref{eq:traceconvergence-nu} and \eqref{eq:traceconvergence-infnu} yields the following notable Corollary:

   \begin{cor}\label{cor:combinedNnu}
    Under the assumptions of Theorem \ref{thm:N-nu-dynamics}, there exists a constant $C>0$, independent of $N$, $\nu$, and $t$ ($N\in\mathbb{N}$ with $N\geqslant 2$, $\nu\geqslant\nu_0$, $t\geqslant 0$), such that
       \begin{equation}\label{eq:fullyquant}
     \mathrm{Tr}\,\big| \,\gamma_{N,\nu,t}^{(1)} - \gamma_{\infty,\infty,t}^{(1)}\,\big|\;\leqslant\;C\,e^{C t}\Big(\frac{1}{N}+\frac{1}{\sqrt{\nu}}\Big)\,.
   \end{equation}
   In particular, in the limit $N,\nu\to\infty$ the one-body marginal $\gamma_{N,\nu,t}^{(1)}$ converges in trace norm to a rank-two density matrix.
   \end{cor}

   With Corollary \ref{cor:combinedNnu} we thus obtain, for the toy model \eqref{eq:nuHilbertspaces}-\eqref{eq:defPsiN2} under consideration, a control on the time evolution of an initial many-body state of the type \eqref{eq:defPsiN2}, with two-level fragmented BEC, with the following two main features:
   \begin{itemize}
    \item whereas in general, at the level of the one-body marginals, $\gamma_{N,\nu,t}^{(1)}$ acquires infinite rank at almost every time $t>0$, and such an infinite rank is present in general also in the limit of infinitely many particles, the additional infinite gap limit (hence, all together, the limit $N,\nu\to\infty$) yields for a.e.~$t>0$ a density matrix with the same rank two initially displayed at $t=0$, and with the explicit Hartree-like effective dynamics \eqref{eq:infgapdynamics1}-\eqref{eq:infgapdynamics2}; this was actually the main result in \cite{DimFalcOlg21-fragmented} (despite the already commented excessively large definition of fragmented BEC adopted therein);
    \item in addition, which was inherently missing in the measure-theoretic analysis of \cite{DimFalcOlg21-fragmented}, the rate of convergence is fully quantitative in $\nu$ \emph{and} $N$, as provided by \eqref{eq:fullyquant} above.
   \end{itemize}

  \section{Proof of Theorem \ref{thm:N-nu-dynamics}}\label{sec:proofofThmdynamics}

  For the proof of Theorem \ref{thm:N-nu-dynamics} the special form of the one-body Hamiltonian and the fact that the spinor sector is not affected by the dynamics play a crucial role.

  Let us start with simple preparatory steps. We observe that, as elements of $\cH=L^2(\mathbb{R}^3)\otimes\mathbb{C}^2$, 
  \begin{equation}\label{eq:phi1phi2phi}
   \varphi_1\,=\,\phi^0\otimes\binom{1}{0}\,,\qquad \varphi_2\,=\,\phi^0\otimes\binom{0}{1}\,,\qquad \textrm{with}\;\;\phi^0(x)\,:=\,(2\pi)^{-\frac{3}{2}}e^{-x^2/2}\,,
   \end{equation}
  and, as an operator on $\cH$,
  \begin{equation}
   \mathfrak{h}_\nu\,=\,h_\nu\otimes\mathbbm{1}\,,\qquad \textrm{with}\;\; h_\nu\,:=\,{\textstyle\frac{1}{2}}\nu(-\Delta+x^2-3)\,.
  \end{equation}
  Therefore, the solution $\varphi_\nu^{\theta_1,\theta_2}$ to the Cauchy problem \eqref{eq:IVPhartree-nu} has at any time $t$ the form
  \begin{equation}
   \varphi_{\nu,t}^{\theta_1,\theta_2}\,=\,\phi_t\otimes\bigg(\sqrt{n_1}\,e^{-\ii\theta_1}\binom{1}{0}+\sqrt{n_2}\,e^{-\ii\theta_2}\binom{0}{1}\bigg)\,=\,\phi_t\otimes\binom{\sqrt{n_1}\,e^{-\ii\theta_1}}{\sqrt{n_2}\,e^{-\ii\theta_2}}\,,
  \end{equation}
 where $\phi\equiv\phi_t(x)$ is the ($\nu$-dependent) solution to the Cauchy problem
  \begin{equation}\label{eq:Hatreeforphit}
   \begin{split}
    \ii\partial_t\phi\:&=\:{\textstyle\frac{1}{2}}\nu(-\Delta+x^2-3)\phi+(V*|\phi|^2)\phi\,, \\
    \phi_{t=0}\:&=\:\phi^0
   \end{split}
  \end{equation}
  (Hartree evolution on the spatial sector only). This also means that \eqref{eq:gammainfk-nu} reads
  \begin{equation}\label{eq:gammainfnutrewritten}
   \gamma_{\infty,\nu,t}^{(k)}\,=\,|\phi_t^{\otimes k}\rangle\langle\phi_t^{\otimes k}|\otimes\frac{1}{\:(2\pi)^2}\iint_{[0,2\pi]^2}\ud\theta_1\,\ud\theta_2 \,\bigg|\binom{\sqrt{n_1}\,e^{-\ii\theta_1}}{\sqrt{n_2}\,e^{-\ii\theta_2}}^{\!\!\otimes k}\bigg\rangle\bigg\langle \binom{\sqrt{n_1}\,e^{-\ii\theta_1}}{\sqrt{n_2}\,e^{-\ii\theta_2}}^{\!\!\otimes k}\bigg|\,,
  \end{equation}
   as an operator acting on (the bosonic sector of) $L^2(\mathbb{R}^{3k})\otimes(\mathbb{C}^{2})^{\otimes k}$ -- see \eqref{eq:isoN} below.

  Analogously, at the $N$-body level it is convenient to exploit the canonical Hilbert space isomorphism
  \begin{equation}\label{eq:isoN}
   (L^2(\mathbb{R}^3)\otimes\mathbb{C}^2)^{\otimes N}\;\cong\;L^2(\mathbb{R}^{3N})\otimes(\mathbb{C}^{2})^{\otimes N}
  \end{equation}
  and re-write 
  \begin{equation}\label{eq:HN-HNspat}
   H_{N,\nu}\:=\;H^{\mathrm{spat}}_{N,\nu}\otimes\mathbbm{1}\,,\qquad\textrm{with}\;\;H^{\mathrm{spat}}_{N,\nu}\::=\:\sum_{j=1}^N h_{\nu,j}\,+\frac{1}{N}\!\!\sum_{1\leqslant \ell<r\leqslant N }\!\!V(x_\ell-x_r)
  \end{equation}
 (`spat' standing for the \emph{spatial} sector, and $h_{\nu,j}$ acting as $h_\nu$ on the $j$-th variable). The initial $N$-body state \eqref{eq:defPsiN2} is re-written as
 \begin{equation}
  \Psi_N\::=\;(\phi^0)^{\otimes N}\otimes \bigg(\binom{1}{0}^{\!\!\otimes N_1}\vee\binom{0}{1}^{\!\!\otimes N_2} \bigg)
 \end{equation}
  and its evolved version is re-written as
  \begin{equation}
  \begin{split}
    \Psi_{N,\nu,t}\;=\;e^{-\ii t H_{N,\nu}}\Psi_N\;&=\;\big(e^{-\ii t H^{\mathrm{spat}}_{N,\nu}}(\phi^0)^{\otimes N}\big)\otimes \bigg(\binom{1}{0}^{\!\!\otimes N_1}\vee\binom{0}{1}^{\!\!\otimes N_2} \bigg) \\
    &\equiv\;\qquad\;\; \Psi_{N,\nu,t}^{\mathrm{spat}}\qquad\;\;\otimes \qquad\qquad\Psi_{N_1,N_2}^{\mathrm{spin}}
  \end{split}
  \end{equation}
  (`spin' standing for the \emph{spinor} sector).
  In turn, this implies that the $k$-body reduced density matrix associated to $ \Psi_{N,\nu,t}$ has the form 
  \begin{equation}\label{eq:factorisationofgamman}
   \gamma_{N,\nu,t}^{(k)}\;=\;\gamma_{\mathrm{spat},N,\nu,t}^{(k)}\otimes\gamma_{\mathrm{spin},N_1,N_2}^{(k)}\,,
  \end{equation}
  where $\gamma_{\mathrm{spat},N,\nu,t}^{(k)}$ is the $k$-marginal of $|\Psi_{N,\nu,t}^{\mathrm{spat}}\rangle\langle \Psi_{N,\nu,t}^{\mathrm{spat}}|$ obtained by tracing out $N-k$ spatial degrees of freedom, and $\gamma_{\mathrm{spin},N_1,N_2}^{(k)}$ is the $k$-marginal of $|\Psi_{N_1,N_2}^{\mathrm{spin}}\rangle\langle\Psi_{N_1,N_2}^{\mathrm{spin}}|$ obtained by tracing out $N-k$ spinor degrees of freedom. The ket-bra notation refers, respectively, to the spatial variables Hilbert space and the spin variables Hilbert space.

  We can prove the following.
  
  \begin{prop}\label{prop:spinor-k-marginal}
   One has
   \begin{equation}\label{eq:spinor-k-marginal}
    \gamma_{\mathrm{spin},N_1,N_2}^{(k)}\;=\;\sum_{j=0}^k\,c_{k,j}\,\bigg|\binom{1}{0}^{\!\!\otimes(k-j)}\vee\binom{0}{1}^{\!\!\otimes j}\bigg\rangle\bigg\langle\binom{1}{0}^{\!\!\otimes(k-j)}\vee\binom{0}{1}^{\!\!\otimes j}\bigg|\,,\qquad c_{k,j}\,:=\,\frac{\,\binom{N_1}{k-j} \binom{N_2}{j}\,}{\binom{N}{k}}\,.
   \end{equation}
  \end{prop}
  It is not by chance that the coefficients $c_{k,j}$ are the very same as in \eqref{eq:gammak}. Indeed, the statement of Proposition \ref{prop:spinor-k-marginal} has the same structure as the implication \eqref{eq:defPsiN}$\Rightarrow$\eqref{eq:gammak} in Proposition \ref{prop:reduced-marginals-time-zero}. Explicitly, to claim that the $k$-marginal of $\Psi_N=\varphi_{1}^{\otimes N_{1}}\vee\varphi_{2}^{\otimes N_{2}}$ is $\gamma_N^{(k)}=\sum_{j=0}^{k}\,c_{k,j}\,
    \big|\varphi_{1}^{\otimes(k-j)}\vee\varphi_{2}^{\otimes j}\big\rangle\big\langle\varphi_{1}^{\otimes(k-j)}\vee\varphi_{2}^{\otimes j}\big|$ is tantamount as to claim that the $k$-marginal of $\Psi_{N_1,N_2}^{\mathrm{spin}}=\binom{1}{0}^{\!\!\otimes N_1}\vee\binom{0}{1}^{\!\!\otimes N_2}$ is 
    $ \gamma_{\mathrm{spin},N_1,N_2}^{(k)}=\sum_{j=0}^k\,c_{k,j}\,\big|\binom{1}{0}^{\!\!\otimes(k-j)}\vee\binom{0}{1}^{\!\!\otimes j}\big\rangle\big\langle\binom{1}{0}^{\!\!\otimes(k-j)}\vee\binom{0}{1}^{\!\!\otimes j}\big|$ up to the canonical isomorphism
    \[
     \mathrm{span}\{\varphi_1,\varphi_2\}^{\otimes N}\,\xrightarrow{\cong}\,(\mathbb{C}^2)^{\otimes N}\,,\qquad \varphi_1\,\mapsto\,\binom{1}{0}\,,\quad \varphi_2\,\mapsto\,\binom{0}{1}\,.
    \]
    Thus, by proving Proposition \ref{prop:spinor-k-marginal} (which will be done in a moment) one also proves Proposition \ref{prop:reduced-marginals-time-zero}.

  With these preparations at hand, we can now control the asymptotic behaviour of both $\gamma_{\mathrm{spat},N,\nu,t}^{(k)}$ and $\gamma_{\mathrm{spin},N_1,N_2}^{(k)}$ as $N\to\infty$.

  \begin{thm}\label{thm:spatialdynamics} Under the assumption of Theorem \ref{thm:N-nu-dynamics}, for every $k\in\{1,\dots,N\}$ and arbitrary $\nu_0>0$ there exists constants $b_k\equiv b_k(\phi^0,\nu_0)>0$,  independent of $N$, $t$, and $\nu\geqslant\nu_0$, such that, for every $t>0$,
  \begin{equation}\label{eq:spatialdynamics}
   \mathrm{Tr}\,\big|\,\gamma_{\mathrm{spat},N,\nu,t}^{(k)}-|\phi_t^{\otimes k}\rangle\langle\phi_t^{\otimes k}|\,\big|\;\leqslant\;b_k\,\frac{e^{b_k t}}{N}\,.
  \end{equation}
  The above trace is in $L^2(\mathbb{R}^{3k})$. The one-body orbital $\phi_t$ is characterised in \eqref{eq:Hatreeforphit}.
  \end{thm}

  \begin{prop}\label{prop:spinlimits}
   Under the assumption of Theorem \ref{thm:N-nu-dynamics} concerning $N_1,N_2,N$, for each $k\in\{1,\dots,N-1\}$ there exists $c_k>0$ such that
   \begin{equation}\label{eq:spinlimits}
    \mathrm{Tr}\,\big|\,\gamma_{\mathrm{spin},N_1,N_2}^{(k)}-\gamma_{\mathrm{spin},\infty}^{(k)}\,\big|\;\leqslant\;\frac{c_k}{N}\,,
   \end{equation}
  where
  \begin{equation}\label{eq:gammaspininf}
   \gamma_{\mathrm{spin},\infty}^{(k)}\;:=\;\sum_{j=0}^k\,c^\infty_{k,j}\,\bigg|\binom{1}{0}^{\!\!\otimes(k-j)}\vee\binom{0}{1}^{\!\!\otimes j}\bigg\rangle\bigg\langle\binom{1}{0}^{\!\!\otimes(k-j)}\vee\binom{0}{1}^{\!\!\otimes j}\bigg|\,,\qquad c^\infty_{k,j}\,:=\,\binom{k}{j} n_1^{k-j}n_2^j\,.
  \end{equation}
   The above trace is in $(\mathbb{C}^{2})^{\otimes k}$.
  \end{prop}

  \begin{prop}\label{prop:gammaspinrewritten}
   For $\gamma_{\mathrm{spin},\infty}^{(k)}$ defined in \eqref{eq:gammaspininf} one also has
   \begin{equation}\label{eq:gammaspinrewritten}
    \gamma_{\mathrm{spin},\infty}^{(k)}\;=\;\frac{1}{\:(2\pi)^2}\iint_{[0,2\pi]^2}\ud\theta_1\,\ud\theta_2 \,\bigg|\binom{\sqrt{n_1}\,e^{-\ii\theta_1}}{\sqrt{n_2}\,e^{-\ii\theta_2}}^{\!\!\otimes k}\bigg\rangle\bigg\langle \binom{\sqrt{n_1}\,e^{-\ii\theta_1}}{\sqrt{n_2}\,e^{-\ii\theta_2}}^{\!\!\otimes k}\bigg|\,.
   \end{equation}
  \end{prop}

  Combining the above results together, Theorem \ref{thm:N-nu-dynamics} follows straightforwardly.
  
  \begin{proof}[Proof of Theorem \ref{thm:N-nu-dynamics}]
   Owing to Theorem \ref{thm:spatialdynamics} and Proposition \ref{prop:spinlimits}, and using \eqref{eq:factorisationofgamman},
   \[\tag{i}
    \mathrm{Tr}\,\Big| \,\gamma_{N,\nu,t}^{(k)}-\Big(|\phi_t^{\otimes k}\rangle\langle\phi_t^{\otimes k}|\otimes \gamma_{\mathrm{spin},\infty}^{(k)}\Big)\,\Big|\;\leqslant\;b_k\frac{e^{b_k t}}{N}+\frac{c_k}{N}\,,
   \]
  the above trace being in the Hilbert space $\cH_N$ defined in \eqref{eq:Hmanybody}, re-written through the canonical isomorphism \eqref{eq:isoN}. Modulo an obvious definition of a new constant $A_k$, the r.h.s.~of (i) can be bounded by $A_k e^{A_k t}/N$, so as to obtain the r.h.s.~of \eqref{eq:traceconvergence-nu}. Concerning the l.h.s., owing to \eqref{eq:gammainfnutrewritten} and Proposition \ref{prop:gammaspinrewritten} we have
  \[\tag{ii}
   |\phi_t^{\otimes k}\rangle\langle\phi_t^{\otimes k}|\otimes \gamma_{\mathrm{spin},\infty}^{(k)}\;=\;\gamma_{\infty,\nu,t}^{(k)}\,.
  \]
  Thus, (i) and (ii) together yield \eqref{eq:traceconvergence-nu}.   
  \end{proof}

  Let us now move on to the proof of the auxiliary results stated above. We defer to the next Section the proof of the effective dynamics on the spatial sector (Theorem \ref{thm:spatialdynamics}), and we focus here on the spinor sector and on the limit of $\gamma_{\mathrm{spin},N_1,N_2}^{(k)}$ as $N\to\infty$ (Propositions \ref{prop:spinor-k-marginal}, \ref{prop:spinlimits}, and  \ref{prop:gammaspinrewritten}).

  \begin{proof}[Proof of Proposition \ref{prop:spinor-k-marginal}]
   Let us assume (non-restrictively) that $N_{1}\leqslant N_{2}$ and introduce the shorthand
   \[
   e_1 \,:=\, \binom{1}{0},\qquad\quad e_2 \,:=\, \binom{0}{1},
   \]
   and let us consider the partial trace in $\gamma_{\mathrm{spin},N_1,N_2}:=|\Psi_{N_1,N_2}^{\mathrm{spin}}\rangle\langle\Psi_{N_1,N_2}^{\mathrm{spin}}|=\big|e_1^{\otimes N_1}\vee e_2^{\otimes N_2}\big\rangle\big\langle e_1^{\otimes N_1}\vee e_2^{\otimes N_2}\big|$ at level $k\in\{1,\dots,N_1\}$. We can conveniently re-arrange the expansion in $e_1^{\otimes N_1}\vee e_2^{\otimes N_2}$ as
   \[
     e_1^{\otimes N_1}\vee e_2^{\otimes N_2}\;=\;\frac{1}{Z}\sum_{j=0}^{N_1}\mathsf{V}_j
   \]
  for a normalisation factor $Z$ that is going to be clear in a moment, and vectors $\mathsf{V}_j\in(\mathbb{C}^2)^{\otimes N_1}\vee (\mathbb{C}^2)^{\otimes N_2}$ defined as follows. $\mathsf{V}_j$ is the \emph{sum} of all terms of the form $\mathsf{w}_j\otimes\widetilde{\mathsf{w}}_j$, where
  \begin{itemize}
   \item $\mathsf{w}_j\in(\mathbb{C}^2)^{\otimes N_1}$ is any tensor product of $N_1-j$ copies of $e_1$ and $j$ copies of $e_2$, that is, $\mathsf{w}_j$ is any of the $\binom{N_1}{j}$ vectors obtained from the representative $e_1^{\otimes (N_1-j)}\otimes e_2^{\otimes j}$ by choosing $j$-out-of-$N_1$ positions for the $e_2$'s and placing the remaining $N_1-j$ vectors of type $e_1$ in the other positions, 
   \item and conversely $\widetilde{\mathsf{w}}_j\in(\mathbb{C}^2)^{\otimes N_2}$ is any of the $\binom{N_2}{j}$ vectors obtained from the representative $e_1^{\otimes j}\otimes e_2^{\otimes (N_2-j)}$ by choosing $j$-out-of-$N_2$ positions for the $e_1$'s and placing the remaining $N_2-j$ vectors of type $e_2$ in the other positions.
  \end{itemize}
  Observe that each $\mathsf{w}_j\otimes\widetilde{\mathsf{w}}_j$ is a $N(=N_1+N_2)$-tensor product with precisely $(N_1-j)+j=N_1$ copies of $e_1$ and $j+(N_2-j)=N_2$ copies of $e_2$, as it must be in the expansion of $e_1^{\otimes N_1}\vee e_2^{\otimes N_2}$. There are by construction $\binom{N_1}{j}\binom{N_2}{j}$ summands in $\mathsf{V}_j$, and they are all mutually orthogonal (due to the different positions occupied by the $e_1$'s and $e_2$'s in any two distinct $\mathsf{w}_j\otimes\widetilde{\mathsf{w}}_j$, and to the orthogonality $e_1\perp e_2$). Moreover, any summand from $\mathsf{V}_j$ is orthogonal to any summand from $\mathsf{V}_{j'}$ for $j\neq j'$ (due to the fact that summands from $\mathsf{V}_j$ and $\mathsf{V}_{j'}$ have a different content of $e_1$'s and $e_2$'s, and again $e_1\perp e_2$). This means that $\sum_{j=0}^{N_1}\mathsf{V}_j$ consists of a sum of normalised $N$-fold tensor products, all mutually orthogonal, whose number amounts to
  \[
\sum_{j=0}^{N_{1}}\binom{N_{1}}{j}\binom{N_{2}}{j}\;=\;\frac{(N_{1}+N_{2})!}{N_{1}!N_{2}!} \;=\; \binom{N}{N_1}\,.
\]
 Thus, the above normalisation factor must be
 \[
  Z\;=\;\binom{N}{N_1}^{\!\frac{1}{2}},
 \]
  and
  \[\tag{i}\label{eq:tobetraced}
   \big|e_1^{\otimes N_1}\vee e_2^{\otimes N_2}\big\rangle\big\langle e_1^{\otimes N_1}\vee e_2^{\otimes N_2}\big|\;=\;\binom{N}{N_1}^{\!-1}\Big|\sum_{j=0}^{N_1}\mathsf{V}_j\Big\rangle\Big\langle \sum_{j'=0}^{N_1}\mathsf{V}_{j'}\Big|\,.
  \]
 We now aim at tracing out $N-k$ degrees of freedom in the density matrix \eqref{eq:tobetraced}, and therefore, by linearity, in each $|\mathsf{V}_j\rangle\langle\mathsf{V}_{j'}|$. In fact, the $k$-partial trace in each $|\mathsf{V}_j\rangle\langle\mathsf{V}_{j'}|$ with $j\neq j'$ gives zero. This is seen for each summand $|\mathsf{w}_j\otimes\widetilde{\mathsf{w}}_j\rangle\langle\mathsf{w}_{j'}\otimes\widetilde{\mathsf{w}}_{j'}|$ arising from the expansion of $|\mathsf{V}_j\rangle\langle\mathsf{V}_{j'}|$: when tracing out the last $N-k$ positions in $|\mathsf{w}_j\otimes\widetilde{\mathsf{w}}_j\rangle\langle\mathsf{w}_{j'}\otimes\widetilde{\mathsf{w}}_{j'}|$, for fixed $k\in\{1,\dots,N_1\}$, one surely has at least one position between the $(k+1)$-th and the $(N_1+N_2)$-th occupied by $e_1$ on the left and $e_2$ on the right, or vice versa, thereby producing zero $k$-marginal owing to the orthogonality $e_1\perp e_2$. The prototypical case 
 \[
 \begin{split}
   &|\mathsf{w}_1\otimes\widetilde{\mathsf{w}}_1\rangle\langle\mathsf{w}_{2'}\otimes\widetilde{\mathsf{w}}_{2'}| \\
   &\qquad=\;\big|\underbrace{e_2  \underbrace{e_1 \cdots  e_1}_{(k-1)} \underbrace{e_1 \cdots  e_1}_{(N_1-k)}}_{(N_1)} \underbrace{e_1 \underbrace{\pmb{e}_{\pmb{2}} \cdots  e_2}_{(k-1)}\underbrace{e_2\otimes\cdots\otimes e_2}_{(N_2-k)}}_{(N_2)}\big\rangle \big\langle \underbrace{e_2e_2\underbrace{e_1\cdots e_1}_{k-2}\underbrace{e_1\cdots e_1}_{N_1-k}}_{(N_1)} \underbrace{e_1\pmb{e}_{\pmb{1}}\underbrace{e_2\cdots e_2}_{(k-2)}\underbrace{e_2\cdots e_2}_{(N_2-k)}}_{(N_2)} \big|
 \end{split}
 \]
(having suppressed the notation `$\otimes$' in the r.h.s.) shows an example of this occurrence (in bold-face one has $\pmb{e}_{\pmb{2}}$ in position $(N_1+2)$ on the left, as opposite to $\pmb{e}_{\pmb{1}}$ in the same position on the right). Therefore,
\[\tag{ii}\label{eq:trnkvj}
  \gamma_{\mathrm{spin},N_1,N_2}^{(k)}\;=\; \mathrm{Tr}_{[N-k]}\big|e_1^{\otimes N_1}\vee e_2^{\otimes N_2}\big\rangle\big\langle e_1^{\otimes N_1}\vee e_2^{\otimes N_2}\big| 
  \;=\; \binom{N}{N_1}^{\!-1}\mathrm{Tr}_{[N-k]}\bigg(\sum_{j=0}^{N_1}\big|\mathsf{V}_j\big\rangle\big\langle \mathsf{V}_{j}\big|\bigg).
\]
 Next, instead of proceeding by linearity as $\mathrm{Tr}_{[N-k]}\big(\sum_{j=0}^{N_1}\big|\mathsf{V}_j\big\rangle\big\langle \mathsf{V}_{j}\big|\big)=\sum_{j=0}^{N_1}\mathrm{Tr}_{[N-k]}\big|\mathsf{V}_j\big\rangle\big\langle \mathsf{V}_{j}\big|$, it is more convenient to reason along this line: from the overall sum $\sum_{j=0}^{N_1}\big|\mathsf{V}_j\big\rangle\big\langle \mathsf{V}_{j}\big|$ the $k$-partial trace produces by definition (see \eqref{eq:pt1}-\eqref{eq:pt4} above) a sum of operators acting on $(\mathbb{C}^2)^{\otimes k}$, all of the form $|\mathbf{a}\rangle\langle\mathbf{b}|$ and with coefficient $+1$ in front of each of them in the sum, where $\mathbf{a},\mathbf{b}\in (\mathbb{C}^2)^{\otimes k}$ are $k$-fold tensor products of $e_1$'s and $e_2$'s. The precise form and amount of such $|\mathbf{a}\rangle\langle\mathbf{b}|$'s is argued as follows. To begin with, the $\mathbf{a}$'s and $\mathbf{b}$'s are all possible tensor products with $k-\ell$ copies of $e_1$ and $\ell$ copies of $e_2$, in all possible dispositions, and for all $\ell$'s running over $\{0,\dots,k\}$. They are precisely the terms produced by the expansion of
 \[\tag{iii}\label{eq:notreadyyet}
  \binom{k}{k-\ell}\,\big|e_1^{\otimes (k-\ell)}\vee e_2^{\otimes \ell}\big\rangle\big\langle e_1^{\otimes (k-\ell)}\vee e_2^{\otimes \ell}\big|\,.
 \]
 The pre-factor above guarantees that the expansion is an actual combination of all possible terms with coefficient $+1$, since the expansion of each $e_1^{\otimes(k-\ell)}\vee e_2^{\otimes \ell}$ has by definition the overall pre-factor $1/\sqrt{\binom{k}{k-\ell}}$. At every fixed $\ell$, each summand of the expansion of \eqref{eq:notreadyyet} appears in the r.h.s.~of \eqref{eq:trnkvj} counted multiple times. Explicitly, tracing out the last $N-k$ degrees of freedom from $\sum_{j=0}^{N_1}\big|\mathsf{V}_j\big\rangle\big\langle \mathsf{V}_{j}\big|$ produces the term \eqref{eq:notreadyyet} as many times as $\binom{N-k}{N_1-k+\ell}$. Indeed, \eqref{eq:notreadyyet} is an expansion of summands $|\mathbf{a}\rangle\langle\mathbf{b}|$ with $k$ positions on each side, occupied by $k-\ell$ $e_1$'s, and by $k$-partial trace of $\sum_{j=0}^{N_1}\big|\mathsf{V}_j\big\rangle\big\langle \mathsf{V}_{j}\big|$ one recovers \eqref{eq:notreadyyet}  as many times as the number of distinct positions occupied by the remaining $N_1-(k-\ell)$ $e_1$'s in the $N-k$ positions that have been traced out. This number is precisely $\binom{N-k}{N_1-k+\ell}$. One could argue alternatively, and equivalently, in terms of the $e_2$'s, thus getting the number $\binom{N-k}{N_2-\ell}$ (and indeed, $\binom{N-k}{N_2-\ell}=\binom{N-k}{N_1-k+\ell}$, since $N_1+N_2=N$). Therefore,
 \[\tag{iv}\label{eq:nowbigsum}
 \begin{split}
    \gamma_{\mathrm{spin},N_1,N_2}^{(k)}\;&=\;\binom{N}{N_1}^{\!-1}\mathrm{Tr}_{[N-k]}\bigg(\sum_{j=0}^{N_1}\big|\mathsf{V}_j\big\rangle\big\langle \mathsf{V}_{j}\big|\bigg) \\
    &=\;\binom{N}{N_1}^{\!-1}\sum_{\ell=0}^k\binom{N-k}{N_1-k+\ell}\binom{k}{k-\ell}\,\big|e_1^{\otimes(k-\ell)}\vee e_2^{\otimes \ell}\big\rangle\big\langle e_1^{\otimes(k-\ell)}\vee e_2^{\otimes \ell}\big|\,.
 \end{split}
 \]
  Since
  \[
   \begin{split}
    \binom{N}{N_1}^{\!-1}\binom{N-k}{N_1-k+\ell}\binom{k}{k-\ell}\;&=\;\frac{N_{1}!N_{2}!}{N!}\cdot\frac{(N-k)!}{(N_{1}-k+\ell)!(N_{2}-\ell)!}\cdot\frac{k!}{\,(k-\ell)!\,\ell!\,} \\
    &=\;\frac{(N-k)!\,k!}{N!}\cdot\frac{N_1!}{(N_{1}-k+\ell)!(k-\ell)!}\cdot\frac{N_2!}{\,(N_2-\ell)!\,\ell!\,} \\
    &=\;\binom{N}{k}^{-1}\binom{N_1}{k-\ell}\binom{N_2}{\ell}\,,
   \end{split}
  \]
 one finally deduces \eqref{eq:spinor-k-marginal} from \eqref{eq:nowbigsum}, thus completing the proof.
\end{proof}

  \begin{proof}[Proof of Proposition  \ref{prop:spinlimits}]
   In view of \eqref{eq:spinor-k-marginal} and \eqref{eq:gammaspininf},
   \[
    \mathrm{Tr}\,\big|\,\gamma_{\mathrm{spin},N_1,N_2}^{(k)}-\gamma_{\mathrm{spin},\infty}^{(k)}\,\big|\;\leqslant\;\sum_{j=0}^k\big|c_{k,j}-c^\infty_{k,j}\big|\,.
   \]
   On account of Lemma \ref{lem:combinatorial}, there is $a_k>0$ such that
   \[
    |c_{k,j}-c^\infty_{k,j}\big|\;\leqslant\;\frac{a_k}{N}
   \]
   uniformly in $j$. The two estimates above combined together yield \eqref{eq:spinlimits} with $c_k := (k+1) a_k$, precisely as for the constant $c_k$ in the proof of Proposition \ref{prop:asymptoticvicinity}.   
  \end{proof}
  
  \begin{proof}[Proof of Proposition \ref{prop:gammaspinrewritten}]
  
  With the shorthand
  \[
   e_1 \,:=\, \binom{1}{0},\qquad\quad e_2 \,:=\, \binom{0}{1},\qquad v_{1}\,:=\,\sqrt{n_{1}}\,e_{1}\,,\qquad v_{2}\,:=\,\sqrt{n_{2}}\,e_{2}
  \]
  one re-writes
\begin{align*}
    \textrm{r.h.s.~of \eqref{eq:gammaspinrewritten}}\;&\equiv\;\frac{1}{(2\pi)^{2}}\iint_{[0,2\pi]^{2}}\ud\theta_{1}\,\ud\theta_{2}\,\bigg|\binom{\sqrt{n_{1}}\,e^{-\ii\theta_{1}}}{\sqrt{n_{2}}\,e^{-\ii\theta_{2}}}^{\!\!\otimes k}\bigg\rangle\bigg\langle\binom{\sqrt{n_{1}}\,e^{-\ii\theta_{1}}}{\sqrt{n_{2}}\,e^{-\ii\theta_{2}}}^{\!\!\otimes k}\bigg|\\
    &=\;\frac{1}{(2\pi)^{2}}\iint_{[0,2\pi]^{2}}\ud\theta_{1}\,\ud\theta_{2}\,\Big|\big( e^{-\ii\theta_{1}}v_{1}+e^{-\ii\theta_{2}}v_{2}\big)^{\otimes k}\Big\rangle\Big\langle \big( e^{-\ii\theta_{1}}v_{1}+e^{-\ii\theta_{2}}v_{2}\big)^{\otimes k}\Big|\,.  
\end{align*}
 Expanding the two terms $\big( e^{-\ii\theta_{1}}v_{1}+e^{-\ii\theta_{2}}v_{2}\big)^{\otimes k}$ in the integrand above and taking all possible mixed products yields expressions of the form $|\cdots\otimes v_1\otimes\cdots\otimes v_2\otimes\cdots\rangle\langle \cdots\otimes v_1\otimes\cdots\otimes v_2\otimes\cdots|$, with all possible positions for the $v_1$'s and the $v_2$'s, each multiplied by factors of the form $e^{-\ii p\theta_1}$ and/or $e^{-\ii q\theta_2}$ for integers $p,q$. Actually, $p=p_A-p_B$ and $q=q_A-q_B$ for a generic term $|A\rangle\langle B|$ that arises from $p_A$ copies of $e^{-\ii\theta_{1}}v_{1}$ from $|A\rangle$ and $p_B$ copies of $e^{-\ii\theta_{1}}v_{1}$ from $\langle B|$, and $q_A$ copies of $e^{-\ii\theta_{2}}v_{2}$ from $|A\rangle$ and  $q_B$ copies of $e^{-\ii\theta_{2}}v_{2}$ from $\langle B|$. Now, whenever $p_A\neq p_B$ (resp., $q_A\neq q_B$), the integration over $\theta_1$ (resp., over $\theta_2$) of $e^{-\ii p\theta_1}$ (resp., of $e^{-\ii q\theta_2}$) gives zero: all such terms do not contribute. Thus, in the integrand above one is only left with the sum of terms $|A\rangle\langle B|$ of the form
 \[\tag{*}\label{eq:bigsumcombinatorial}
  \Big|\bigotimes_{j\in J_A} v_1 \otimes \bigotimes_{\ell\in L_A} v_2\Big\rangle\Big\langle \bigotimes_{m\in J_B} v_1 \otimes \bigotimes_{n\in L_B} v_2\Big|
 \]
 for all possible choices of index sets $J_A,L_A,J_B,L_B\subset\{1,\dots,k\}$ such that
  \begin{align*}
   &J_A\cap L_A=\emptyset\,,\quad J_A\cup L_A=\{1,\dots,k\}\,, \\
   &J_B\cap L_B=\emptyset\,,\quad J_B\cup L_B=\{1,\dots,k\}\,, \\
   &|J_A|=|J_B|\,,\quad\; |L_A|=|L_B|\,.
  \end{align*}
 In such a sum the overall $(2\pi)^{-2}\iint_{[0,2\pi]^{2}}\ud\theta_{1}\,\ud\theta_{2}$ integration gives trivially 1.
 This means that the r.h.s.~of \eqref{eq:gammaspinrewritten} is the sum of all possible terms \eqref{eq:bigsumcombinatorial} (i.e., their linear combination with coefficient $+1$ in front of each summand).
 In turn, one can suitably re-group the above-mentioned summands and write
 \[
  \frac{1}{(2\pi)^{2}}\iint_{[0,2\pi]^{2}}\ud\theta_{1}\,\ud\theta_{2}\,\bigg|\binom{\sqrt{n_{1}}\,e^{-\ii\theta_{1}}}{\sqrt{n_{2}}\,e^{-\ii\theta_{2}}}^{\!\!\otimes k}\bigg\rangle\bigg\langle\binom{\sqrt{n_{1}}\,e^{-\ii\theta_{1}}}{\sqrt{n_{2}}\,e^{-\ii\theta_{2}}}^{\!\!\otimes k}\bigg|\;=\;\sum_{j=0}^k \mathcal{C}_j\,,
 \]
  where $\mathcal{C}_j$ consists of the sum of only those terms \eqref{eq:bigsumcombinatorial} for which $|J_A|=|J_B|=k-j$, and consequently $|L_A|=|L_B|=j$, i.e., those terms \eqref{eq:bigsumcombinatorial} having $k-j$ vectors $v_1$ and $j$ vectors $v_2$ both on the left and on the right of the structure $|A\rangle\langle B|$, in all possible positions of the corresponding tensor products. This means precisely that
    \[
   \mathcal{C}_j\;=\;
   \binom{k}{j}\big| v_1^{\otimes(k-j)}\vee v_2^{\otimes j}\big\rangle\big\langle v_1^{\otimes(k-j)}\vee v_2^{\otimes j} \big|\,,
  \]
  the $\binom{k}{j}$-pre-factor cancelling the normalisation $\binom{k}{j}^{-\frac{1}{2}}$ arising from the expansion of $v_1^{\otimes(k-j)}\vee v_2^{\otimes j}$ both on the left and on the right.
  Thus,
  \[
   \begin{split}
    \frac{1}{(2\pi)^{2}}\iint_{[0,2\pi]^{2}}&\ud\theta_{1}\,\ud\theta_{2}\,\bigg|\binom{\sqrt{n_{1}}\,e^{-\ii\theta_{1}}}{\sqrt{n_{2}}\,e^{-\ii\theta_{2}}}^{\!\!\otimes k}\bigg\rangle\bigg\langle\binom{\sqrt{n_{1}}\,e^{-\ii\theta_{1}}}{\sqrt{n_{2}}\,e^{-\ii\theta_{2}}}^{\!\!\otimes k}\bigg|\;=\;\sum_{j=0}^k \mathcal{C}_j \\
    &=\;\sum_{j=0}^k \binom{k}{j}\big| v_1^{\otimes(k-j)}\vee v_2^{\otimes j}\big\rangle\big\langle v_1^{\otimes(k-j)}\vee v_2^{\otimes j} \big| \\
    &=\;\sum_{j=0}^k \binom{k}{j}\,n_1^{k-j} n_2^j\,\big| e_1^{\otimes(k-j)}\vee e_2^{\otimes j}\big\rangle\big\langle e_1^{\otimes(k-j)}\vee e_2^{\otimes j} \big|\;=\;\,\gamma_{\mathrm{spin},\infty}^{(k)}\,,
   \end{split}
  \]
  having used \eqref{eq:gammaspininf} in the last step.
  \end{proof}

  \section{Effective mean-field dynamics for spatial degrees of freedom}\label{sec:effectiveMFdynamics}

  We complete the line of reasoning started in Section \ref{sec:scenarioINfinitegap} by proving Theorem \ref{thm:spatialdynamics} concerning the quantitative emergence of the asymptotic mean-field dynamics \eqref{eq:Hatreeforphit} and \eqref{eq:spatialdynamics}.

  This is a question of effective many-body dynamics in the Hilbert space $L^2_\mathrm{sym}(\mathbb{R}^{3N})$ with totally factorised (uncorrelated) initial state
  \begin{equation}
   \Psi_{N,\nu,t=0}^{\mathrm{spat}}\;=\;(\phi^0)^{\otimes N}
  \end{equation}
 ($\phi^0$ being defined in \eqref{eq:phi1phi2phi}) and evolution governed by the mean-field Hamiltonian $H^{\mathrm{spat}}_{N,\nu}$ given by \eqref{eq:HN-HNspat}. As such, it can be answered by means of a variety of very sophisticated techniques that have been developed over the last two decades (with also precursors from classical kinetic theory) for the dynamics of \emph{simple} BEC. In this respect, we already mentioned in Sections \ref{sec:compositeBEC} and \ref{sec:Evolutiongeneral} the works \cite{BEGMY-2002,EY-2001,ESY-2006,ESYinvent,RS-2007,ESY-2008,kp-2009-cmp2010,Pickl-JSP-2010,Pickl-LMP-2011,chen-lee-2010-JMP2011,Chen-Lee-Schlein-2011,Pickl-RMP-2015,Benedikter-DeOliveira-Schlein_QuantitativeGP-2012_CPAM2015,Boccato-Cenatiempo-Schlein-2015_AHP2017_fluctuations,Chen-Lee-Lee-JMP2018-rateHartree,Lee-JSP2019-timedepRate,Bossman-Pavlovic-Pickl-Soffer-2020,NNRT-2020} as main representatives of such an ample spectrum of techniques and results.

  For the model \eqref{eq:nuHilbertspaces}-\eqref{eq:defPsiN2} under consideration we shall employ the scheme of \cite{Chen-Lee-Schlein-2011}, consisting of a control of the fluctuations around the mean-field leading dynamics for marginals in a Fock space framework, with optimal rate of convergence, by means of an a-priori bound on the growth of the kinetic energy with respect to an approximate dynamics with quadratic generator.

  In this Section we discuss the two main adaptations needed from \cite{Chen-Lee-Schlein-2011}, namely the insertion of a trapping potential, which is only alluded to but not worked out in \cite[Remark 2]{Chen-Lee-Schlein-2011}, specifically a \emph{harmonic} potential, and the quantitative control of the rate of convergence in terms of the new parameter $\nu$ present here.

 We shall establish Theorem \ref{thm:spatialdynamics} in a slightly more general setting. Let us start with simple but crucial bounds on the underlying Hartree dynamics.

 \begin{thm}\label{thm:Hartreebounds}
  Let $U\in C^\infty(\mathbb{R}^3,\mathbb{R})$ such that $U\geqslant 0$ and $D^\alpha U\in L^\infty(\mathbb{R}^3)$ for all multi-indices $\alpha$ with $|\alpha|\geqslant 2$, and let $V:\mathbb{R}^3\to\mathbb{R}$ be a measurable function such that $V$ is even-symmetric, $-\Delta+U$ self-adjoint on $L^2(\mathbb{R}^3)$ with quadratic form domain
  \begin{equation}
   \mathcal{Q}\;:=\;\Big\{ f\in L^2(\mathbb{R}^3)\,\Big|\,\|f\|_{\mathcal{Q}}\,:=\,\Big(\|f\|^2_{H^1}+\langle f,Uf\rangle_{L^2}\Big)^{\!\frac{1}{2}}<+\infty\Big\}\,,
  \end{equation}
  and
  \begin{equation}\label{eq:VboundU}
   V^2\;\lesssim\;\mathbbm{1}-\Delta+U
  \end{equation}
  on $\mathcal{Q}$, in the sense of quadratic forms. For $\phi^0\in \mathcal{Q}$ and $\nu>0$, there is a unique solution $\phi\equiv\phi_t(x)$, $(t,x)\in[0,+\infty)\times\mathbb{R}^3$, with 
  \begin{equation}
   \phi\,\in\,C([0,+\infty),\mathcal{Q})\cap C^1([0,+\infty),\mathcal{Q}^*)
  \end{equation}
  ($\mathcal{Q}^*$ being the topological dual of $Q$ with respect to the norm topology induced by $\|\cdot\|_{\mathcal{Q}}$) to the initial value problem
  \begin{equation}\label{eq:Hartreeharmonic}
   \begin{split}
    \ii\partial_t\phi\,&=\,\nu(-\Delta_x)\phi+\nu\,U\phi+(V*|\phi|^2)\phi\,, \\
    \phi_{t=0}\,&=\,\phi^0
   \end{split}
  \end{equation}
  (the convolution being with respect to the $x$-variable). Moreover, for $\nu$ large enough in terms of the given $U,V,\phi^0$, 
  \begin{equation}\label{eq:unifbound}
   \|\phi_t\|^2_{H^1}+\langle\phi_t,U\phi_t\rangle_{L^2}\;\lesssim\;\|\phi^0\|^2_{H^1}+\langle\phi^0,U\phi^0\rangle_{L^2}
  \end{equation}
  (i.e., $\|\phi_t\|_{\mathcal{Q}}\lesssim\|\phi^0\|_{\mathcal{Q}}$) uniformly in time.  
 \end{thm}

 \begin{proof}
  Standard analysis (see, e.g., \cite[Section 9.2]{cazenave}) establishes the well-posedness of \eqref{eq:Hartreeharmonic} with conservation in time of mass and energy
  \[
   \mathcal{M}(\phi_t)\;:=\;\|\phi_t\|^2_{L^2}\,,\qquad\mathcal{E}(\phi_t)\;:=\;\frac{\nu}{2}\int_{\mathbb{R}^3}|\nabla_x\phi_t|^2+\frac{\nu}{2}\int_{\mathbb{R}^3}U|\phi_t|^2+\frac{1}{4}\int_{\mathbb{R}^3}(V*|\phi_t|^2)|\phi_t|^2\,.
  \]
  From
  \[
   \begin{split}
    \|\phi_t\|_{H^1}^2\;&=\;\|\phi_t\|_{L^2}^2+\|\nabla_x\phi_t\|_{L^2}^2 \\
    &=\;\mathcal{M}(\phi_t)+\frac{2}{\nu}\,\mathcal{E}(\phi_t)-\int_{\mathbb{R}^3}U|\phi_t|^2-\frac{1}{2\nu}\int_{\mathbb{R}^3}(V*|\phi_t|^2)|\phi_t|^2
   \end{split}
  \]
  one finds
  \[
   \|\phi_t\|_{H^1}^2+\int_{\mathbb{R}^3}U|\phi_t|^2\;\leqslant\;\mathcal{M}(\phi^0)+\frac{2}{\nu}\,\mathcal{E}(\phi^0)+\frac{1}{2\nu}\int_{\mathbb{R}^3}|V*|\phi_t|^2|\,|\phi_t|^2\,,
  \]
  and from \eqref{eq:VboundU} one finds
  \[
   \begin{split}
    \big\| V*|\phi_t|^2 \big\|_{L^\infty}\;& \leqslant\;\frac{1}{2}\|\phi_t\|_{L^2}^2+\frac{1}{2}\,\mathrm{ess}\sup_{\!\!\!\!\!\!\!\!x}\int_{\mathbb{R}^3}V^2(x-y)|\phi_t(y)|^2\,\ud y  \\
    &\;\leqslant \frac{1}{2}\,\mathcal{M}(\phi^0)+\frac{\kappa}{2} \,\Big( \|\phi_t\|_{H^1}^2+\int_{\mathbb{R}^3}U|\phi_t|^2\Big)
   \end{split}
 \]
 for some $\kappa>0$ depending only on the given $U$ and $V$. Therefore, for sufficiently large $\nu>0$, so as $\frac{\,\kappa\mathcal{M}(\phi^0)}{4\nu}<1$,
 \[
  \|\phi_t\|_{H^1}^2+\int_{\mathbb{R}^3}U|\phi_t|^2\;\leqslant\;\Big(1-\frac{\,\kappa\mathcal{M}(\phi^0)}{4\nu}\Big)^{-1}\Big(\mathcal{M}(\phi^0)+\frac{\,\mathcal{M}(\phi^0)^2}{4\nu}+\frac{2}{\nu}\,\mathcal{E}(\phi^0)\Big)\,.
 \]
  Observe that
  \[
   \Big(1-\frac{\,\kappa\mathcal{M}(\phi^0)}{4\nu}\Big)^{-1}\Big(\mathcal{M}(\phi^0)+\frac{\,\mathcal{M}(\phi^0)^2}{4\nu}+\frac{2}{\nu}\,\mathcal{E}(\phi^0)\Big)\;\xrightarrow{\;\nu\to +\infty\;}\;\|\phi^0\|_{H^1}^2+\int_{\mathbb{R}^3}U|\phi^0|^2\,,
  \] 
  implying that for all $\nu$'s exceeding a (large) threshold $\nu_0$ which depends on $U,V,\phi^0$ only, one has
  \[
   \|\phi_t\|_{H^1}^2+\int_{\mathbb{R}^3}U|\phi_t|^2\;\leqslant\; C_0\Big( \|\phi^0\|_{H^1}^2+\int_{\mathbb{R}^3}U|\phi^0|^2\Big) 
  \]
 uniformly in time for some $C_0>1$, which is the estimate \eqref{eq:unifbound}.
 \end{proof}

 The main catch from Theorem \ref{thm:Hartreebounds} is the uniform boundedness in time of both $\|\phi_t\|_{H^1}$ and $\langle\phi_t,U\phi_t\rangle_{L^2}$, controlled by the constant $\|\phi^0\|_{H^1}+\langle\phi^0,U\phi^0\rangle_{L^2}$.
 Concerning the well-posedness claim, the same standard arguments \cite[Section 9.2]{cazenave} obviously yield well-posedness in $\mathcal{Q}$ of the modified problem \eqref{eq:Hartreeharmonic-ab} below.

 Here is the slightly more general statement from which we deduce Theorem \ref{thm:spatialdynamics}.

\begin{thm} \label{thm:CLSmain}

Let $U\in C^\infty(\mathbb{R}^3,\mathbb{R})$ such that $U\geqslant 0$ and $D^\alpha U\in L^\infty(\mathbb{R}^3)$ for all multi-indices $\alpha$ with $|\alpha|\geqslant 2$, and let $V:\mathbb{R}^3\to\mathbb{R}$ be a measurable function such that $V$ is even-symmetric, $-\Delta+U$ self-adjoint on $L^2(\mathbb{R}^3)$ with quadratic form domain
  \begin{equation}\label{eq:Qnorm}
   \mathcal{Q}\;:=\;\Big\{ f\in L^2(\mathbb{R}^3)\,\Big|\,\|f\|_{\mathcal{Q}}\,:=\,\Big(\|f\|^2_{H^1}+\langle f,Uf\rangle_{L^2}\Big)^{\!\frac{1}{2}}<+\infty\Big\}\,,
  \end{equation}
  and
  \begin{equation}\label{eq:Vcond}
   V^2\;\leqslant\;\mathsf{q}(\mathbbm{1}-\Delta+U)
  \end{equation}
  on $\mathcal{Q}$, in the sense of quadratic forms, for some $\mathsf{q}>0$.
  Correspondingly, and for $N\in\mathbb{N}$ with $N\geqslant 2$, let 
  \begin{equation}
 \mathsf{H}_{N,\nu }\;:=\;\sum_{j=1}^N\big(\nu (-\Delta_{x_j})+\nu \,U(x_j)\big)+\frac{1}{N}\sum_{1\leqslant\ell<r\leqslant N}V(x_\ell-x_r)
  \end{equation}
  be the self-adjoint operator on $L^2_\mathrm{sym}(\mathbb{R}^{3N})$ with domain of essential self-adjointness $C^\infty_c(\mathbb{R}^{3N})\cap L^2_\mathrm{sym}(\mathbb{R}^{3N})$, let $\phi\equiv\phi_t(x)$, $(t,x)\in\mathbb{R}\times\mathbb{R}^3$, be the unique solution  
  \begin{equation}
   \phi\,\in\,C([0,+\infty),\mathcal{Q})\cap C^1([0,+\infty),\mathcal{Q}^*)
  \end{equation}
   to the initial value problem
  \begin{equation}\label{eq:Hartreeharmonic-ab}
   \begin{split}
    \ii\partial_t\phi\,&=\,\nu (-\Delta_x)\phi+\nu \,U\phi+(V*|\phi|^2)\phi\,, \\
    \phi_{t=0}\,&=\,\phi^0\,,
   \end{split}
  \end{equation}
  and for each $t\geqslant 0$ and $k\in\{1,\dots,N-1\}$ let
  \begin{equation}\label{eq:psinabt}
   \Psi_{N,\nu ,t}\;:=\;e^{-\ii t \mathsf{H}_{N,\nu }}(\phi^0)^{\otimes N}\;\in\;L^2_\mathrm{sym}(\mathbb{R}^{3N})
  \end{equation}
  and 
  \begin{equation}
   \gamma_{N,\nu ,t}^{(k)}\;:=\;\mathrm{Tr}_{[N-k]}|\Psi_{N,\nu ,t}\rangle\langle\Psi_{N,\nu ,t}|\,.
  \end{equation}
  Then, for arbitrary $\nu_0>0$, there are positive constants
  \begin{equation}\label{eq:quantities}
   C_k\;\equiv\;C_k\big(\mathsf{q},\| \phi^0 \|_{\mathcal{Q}},\nu_0\big)
  \end{equation}
  such that 
  \begin{equation}\label{eq:main}
   \mathrm{Tr}\,\big|\,\gamma_{N,\nu ,t}^{(k)}-|\phi_t^{\otimes k}\rangle\langle\phi_t^{\otimes k}| \,\big|\;\leqslant\;\frac{\,C_k \,e^{ t C_k}}{N}
  \end{equation}
  for each $N,t,k$ as above, and each $\nu\geqslant\nu_0$. 
\end{thm}

  Obviously, the symbols $\mathsf{H}_{N,\nu }$, $\Psi_{N,\nu ,t}$, and $\gamma_{N,\nu ,t}^{(k)}$ in this Section correspond to $H_{N,\nu}^{\mathrm{spat}}$, $\Psi_{N,\nu,t}^{\mathrm{spat}}$, and $\gamma_{\mathrm{spat},N,\nu,t}^{(k)}$ used in Section \ref{sec:proofofThmdynamics}. Since only spatial variables are considered in the present Section, the redundant `spat' label is omitted here with no risk of confusion.

  It is also clear that Theorem \ref{thm:spatialdynamics} is a special case of Theorem \ref{thm:CLSmain} with the choice $U(x)=x^2$ (plus an irrelevant shift by an additive constant) and $\phi^0=(2\pi)^{-3/2}e^{-x^2/2}$ (so that $\phi_t$ in Theorem \ref{thm:CLSmain} is precisely $\phi_t$ considered in Theorem \ref{thm:spatialdynamics}).

   In the remaining part of this Section we come to the proof of Theorem \ref{thm:CLSmain}.

   \begin{proof}[Proof of Theorem \ref{thm:CLSmain}]
    As already mentioned at the beginning of this Section, this is an adaptation of the analysis \cite{Chen-Lee-Schlein-2011} to the modified Hamiltonian $\mathsf{H}_{N,\nu }$ that includes here a trapping potential and the parameter $\nu$. To this aim one has to revisit all those steps of the scheme of \cite{Chen-Lee-Schlein-2011} where the quantities $U$ and $\nu$ were not originally accounted for, thereby tracking down the various estimates that are going to produce the final constants \eqref{eq:quantities} in the bound \eqref{eq:main}. We only explicitly work out such steps in the following, while referring to the very clean and accessible presentation of \cite{Chen-Lee-Schlein-2011} for all other aspects, that remain untouched, of the general derivation of the bound \eqref{eq:main}. The reader who is not already familiar with \cite{Chen-Lee-Schlein-2011} would find the proof complete by simply replacing our modifications to \cite{Chen-Lee-Schlein-2011} in the precise points indicated each time here below.

\medskip

\textbf{Regularisation of the interaction.}
As in the original scheme \cite{Chen-Lee-Schlein-2011}, a suitable regularisation of the (possibly singular) interaction $V$ is implemented by setting, for concreteness,
\begin{equation}\label{eq:wtV}
\widetilde V (x) \;:=\;\text{sgn}(V (x)) \cdot \min \{ |V (x)|, N^3 \}\,.
\end{equation}
One therefore replaces in \cite{Chen-Lee-Schlein-2011} the regularised Hamiltonian and the regularised evolution with, respectively,
  \begin{equation}\label{eq:regH}
 \wt{\mathsf{H}}_{N,\nu }\;:=\;\sum_{j=1}^N\big(\nu (-\Delta_{x_j})+\nu \,U(x_j)\big)+\frac{1}{N}\sum_{1\leqslant\ell<r\leqslant N}\wt{V}(x_\ell-x_r)
  \end{equation}
and 
\begin{equation}\label{eq:psinabtREG}
   \wt{\Psi}_{N,\nu ,t}\;:=\;e^{-\ii t \wt{\mathsf{H}}_{N,\nu }}(\phi^0)^{\otimes N}\;\in\;L^2_\mathrm{sym}(\mathbb{R}^{3N})\,.
\end{equation}
Full and regularised marginals are denoted by $\gamma_{N,\nu ,t}^{(k)}$ and $\widetilde{\gamma}_{N,\nu ,t}^{(k)}\,$, and next to the full one-body effective dynamics \eqref{eq:Hartreeharmonic} one also considers the regularised initial value problem 
\begin{equation}\label{eq:hartree-reg} 
   \begin{split}
    \ii\partial_t\widetilde{\phi}\,&=\,\nu (-\Delta_x)\widetilde{\phi}+\nu \,U\widetilde{\phi}+(\widetilde{V}*|\widetilde{\phi}|^2)\widetilde{\phi}\,, \\
    \widetilde{\phi}_{t=0}\,&=\,\phi^0\,.
   \end{split}
  \end{equation}

  \textbf{Vicinity between full and regularised dynamics.} Proceeding as in \cite{Chen-Lee-Schlein-2011}, one then controls the large-$N$ asymptotics of $\gamma_{N,\nu ,t}^{(k)}$ by splitting
  \begin{equation}\label{eq:vicinityCLS}
\begin{split}
   &\mathrm{Tr}\,\big|\,\gamma_{N,\nu ,t}^{(k)}-|\phi_t^{\otimes k}\rangle\langle\phi_t^{\otimes k}| \,\big|\\
   &\leqslant\;
   \mathrm{Tr}\,\big|\,\gamma_{N,\nu ,t}^{(k)}-\wt{\gamma}_{N,\nu ,t}^{(k)}\big|\,
   +\mathrm{Tr}\,\big|\,\wt{\gamma}_{N,\nu ,t}^{(k)}-|\wt{\phi}_t^{\otimes k}\rangle\langle\wt{\phi}_t^{\otimes k}|\big|\,
   +\mathrm{Tr}\,\big|\,|\wt{\phi}_t^{\otimes k}\rangle\langle\wt{\phi}_t^{\otimes k}|-|\phi_t^{\otimes k}\rangle\langle\phi_t^{\otimes k}|\big|\,.
\end{split}
\end{equation}
  The vectors \eqref{eq:psinabt} and \eqref{eq:psinabtREG} are as close as
\begin{equation}\label{eq:vicinity-psi-psitilde}
\big\| {\Psi}_{N,\nu ,t} - \wt{\Psi}_{N,\nu ,t} \big\|_{\cH_N}^2 \;\leqslant\; \mathsf{C}_k^2 N^{-2} \|\phi^0\|_{\mathcal{Q}}^2\, t \,
\end{equation}
 for some ($N$- and $t$-independent) constant $\mathsf{C}_k>0$. This is seen with the very same argument of \cite[Lemma 2.1]{Chen-Lee-Schlein-2011}: the now present external potential terms are cancelled out in the commutator $[{\mathsf{H}}_{N,\nu}-\widetilde{\mathsf{H}}_{N,\nu},\,\cdot\,]$, which is the key tool at the basis of that result -- see \cite[formula (2.5)]{Chen-Lee-Schlein-2011}. Moreover, implementing in the proof of \cite[Lemma 2.1]{Chen-Lee-Schlein-2011} the control of $V$ in terms of $-\Delta$ and $V$ given by \eqref{eq:Vcond} above, and the bound \eqref{eq:unifbound}, do yield $\|\phi^0\|_{\mathcal{Q}}$ in the final estimate \eqref{eq:vicinity-psi-psitilde}, which still incorporates the information of the external potential. This also means that $\mathsf{C}_k$ depends on the parameter $\mathsf{q}$ from \eqref{eq:Vcond}. The above vicinity of $N$-body vectors implies (see \eqref{eq:vicinities}) trace norm vicinity of the corresponding $k$-marginals, with the same bound, that is,
 \begin{equation}\label{eq:firsttraceterm}
\Tr \, \big| \gamma_{N,\nu ,t}^{(k)}-\wt{\gamma}_{N,\nu ,t}^{(k)} \big| \;\leqslant\;  \mathsf{C}_k N^{-1} \|\phi^0\|_{\mathcal{Q}}\, t^{\frac{1}{2}}\,.
\end{equation}
 The same cancellation of the trapping potential $U$ occurs when monitoring the difference $\phi_t - \wt{\phi}_t$ between the two initial value problems \eqref{eq:Hartreeharmonic} and \eqref{eq:hartree-reg}: thus, reasoning precisely as in \cite[Lemma 2.2]{Chen-Lee-Schlein-2011} and using now the bound \eqref{eq:unifbound} one obtains
\begin{equation}\label{eq:ph-wtph}
\| \phi_t - \wt{\phi}_t \|_{L^2} \;\leqslant\; \mathsf{B} N^{-\frac{3}{2}}\|\phi^0\|_{\mathcal{Q}}\,e^{\mathsf{B} t}  
\end{equation}
for some ($N$- and $t$-independent) constant $\mathsf{B}>0$ (depending of $\mathsf{q}$ via \eqref{eq:Vcond}), and consequently 
\begin{equation}\label{eq:gm-wtgm} 
\Tr \; \big| |\phi_t \rangle \langle \phi_t |^{\otimes k} - |\wt{\phi}_t \rangle \langle \wt{\phi}_t |^{\otimes k} \big|\; \leqslant\; 2 \| \phi_t - \wt{\phi}_t \|_{L^2} \;\leqslant\; 2\,\mathsf{B} N^{-\frac{3}{2}}\|\phi^0\|_{\mathcal{Q}}\,e^{\mathsf{B} t}\,. 
\end{equation} 
We combine \eqref{eq:firsttraceterm} and \eqref{eq:gm-wtgm} and deduce from \eqref{eq:vicinityCLS}
\begin{equation}\label{eq:whatremains}
 \mathrm{Tr}\,\big|\,\gamma_{N,\nu ,t}^{(k)}-|\phi_t^{\otimes k}\rangle\langle\phi_t^{\otimes k}| \,\big|\;\leqslant\;\frac{\,\mathsf{A}_k\,e^{\mathsf{A}_k t}}{N}+\mathrm{Tr}\,\big|\,\wt{\gamma}_{N,\nu ,t}^{(k)}-|\wt{\phi}_t^{\otimes k}\rangle\langle\wt{\phi}_t^{\otimes k}|\big|
\end{equation}
 for constants $\mathsf{A}_k\equiv\mathsf{A}_k(\mathsf{q},\|\phi^0\|_{\mathcal{Q}})$.

\medskip

\textbf{Effective regularised dynamics.} To control $\mathrm{Tr}|\,\wt{\gamma}_{N,\nu ,t}^{(k)}-|\wt{\phi}_t^{\otimes k}\rangle\langle\wt{\phi}_t^{\otimes k}||$ in \eqref{eq:whatremains}, one adapts \cite[Proposition 2.1]{Chen-Lee-Schlein-2011} to the current setting. In that scheme, two Fock space dynamics are to be compared, the dynamics of fluctuations with complete content of creation and annihilation operators, and the fluctuation dynamics without cubic and quartic terms. Due to the presence in \eqref{eq:regH} of the external potential $U$ and of the parameter $\nu $, the generators $\mathcal{L}$ and $\mathcal{L}_2$, respectively, of such two dynamics take the form
\begin{equation}\label{eq:cL}
\begin{split} 
\cL (t) \;=\; \; & \nu  \int \ud x \,  \nabla_x a^*_x \nabla_x a_x + \nu  \int \ud x \,  U(x) \, a^*_x a_x + \int \ud x \, (\wt V*|\wt\phi_t|^2 ) (x)\, a^*_x a_x \\ 
&+ \int \ud x \ud y \, \wt V(x-y) a_x^* a_y \wt\phi_t (x) \overline{\wt\phi_t (y)}\\
&+ \int \ud x \ud y \, \wt V(x-y) \left( a_x^* a_y^* \, \wt \phi_t (x) \wt \phi_t (y) +  a_x a_y \, \overline{\wt \phi_t (x)}\, \overline{\wt \phi_t (y)} \right) \\
&+\frac{1}{\sqrt{N}} \int \ud x \ud y \, \wt V(x-y) a_x^* \left( a_y^* \wt \phi_t (y) + a_y \overline{\wt \phi_t (y)} \right) a_x \\
&+\frac{1}{N} \int \ud x \ud y \, \wt V(x-y) a_x^* a_y^* a_y a_x 
\end{split}
\end{equation}
and
\begin{equation}\label{eq:L2} 
\begin{split} 
\cL_2 (t) \;=\; \; & \nu  \int \ud x \, \nabla_x a^*_x \nabla_x a_x + \nu  \int \ud x \,  U(x) \, a^*_x a_x + \int \ud x \,  (\wt V*|\wt \phi_t|^2 ) (x) a^*_x a_x \\
&+ \int \ud x \ud y \, \wt V(x-y) a_x^* a_y \wt \phi_t (x) \overline{\wt \phi_t (y)} \\
&+ \int \ud x \ud y \, \wt V(x-y) \left( a_x^* a_y^* \, \wt \phi_t (x) \wt \phi_t (y) +  a_x a_y \, \overline{\wt \phi_t (x)}\, \overline{\wt \phi_t (y)}\right) \,. 
\end{split}
\end{equation}
They uniquely identify the propagators $\cU$ and $\cU_2$ via, respectively,
\begin{align}
\ii \frac{\ud}{\ud t}  \cU (t;s) \,&=\, \cL (t)\, \cU (t;s)\,, \qquad \quad \cU (s;s) = \mathbbm{1} \qquad\,\,\forall\,t,s\in\mathbb{R}\,,\label{eq:U} \\
\ii \frac{\ud}{\ud t}  \cU_2 (t;s) \,&=\, \cL_2 (t)\, \cU_2 (t;s)\,, \,\quad \quad \cU_2 (s;s) = \mathbbm{1}  \qquad\forall\,t,s\in\mathbb{R}\label{eq:U2}
\end{align}
 ($\mathbbm{1}$ denoting here the identity operator on the Fock space).
 With this modification, \cite[Proposition 2.1]{Chen-Lee-Schlein-2011} can be re-proved, 
 so as to finally obtain the bound
 \begin{equation}\label{trace norm bound}
  \mathrm{Tr}\,\big|\,\wt{\gamma}_{N,\nu ,t}^{(k)}-|\wt{\phi}_t^{\otimes k}\rangle\langle\wt{\phi}_t^{\otimes k}|\big|\;\leqslant\;\frac{\,\mathsf{D}_k \,e^{\mathsf{D}_kt}\,}{N}
 \end{equation}
 for some ($N$- and $t$-independent) constants $\mathsf{D}_k\equiv\mathsf{D}_k(\mathsf{q},\|\phi^0\|_{\mathcal{Q}}, \nu)>0$ appearing in \eqref{trace norm bound}, each of which is uniformly bounded in $\nu$ for $\nu\geqslant\nu_0$, with arbitrary $\nu_0>0$. This requires to adapt to the present setting also the preparatory results devised in \cite{Chen-Lee-Schlein-2011} to obtain their original version of \eqref{trace norm bound}: this is going to be discussed here below. Finally, plugging \eqref{trace norm bound} into \eqref{eq:whatremains} yields the desired estimate \eqref{eq:quantities}-\eqref{eq:main}.
 
  \medskip

\textbf{Bounds on the growth of number of particles.} Several technical results that are crucial for the proof of \cite[Proposition 2.1]{Chen-Lee-Schlein-2011} need be revisited as well, and adapted to the present setting, in order to be applicable to the modified version discussed above which leads to the bound \eqref{trace norm bound}. The first one is the control of the  growth of the number of particles with respect to the evolutions $\cU$ and $\cU_2$. To this aim, the proof of \cite[Proposition 3.3]{RS-2007} can be repeated step by step, now with $\cU$ and $\cU_2$ given by \eqref{eq:cL}-\eqref{eq:L2} and \eqref{eq:U}-\eqref{eq:U2}, with the assumption \eqref{eq:Vcond} on the interaction potential, and with the estimate \eqref{eq:unifbound}. This yields the bounds
\begin{align}
  \| (\cN+\mathbbm{1})^j \,\cU (t;s) \,\psi \|_{\mathcal{F}} \;&\leqslant\; \mathsf{K}_j \; e^{\mathsf{K}_j |t-s|} \| (\cN +\mathbbm{1})^{2j+1} \psi \|_{\mathcal{F}}\,, \label{eq:boundgrowth1} \\
  \| (\cN+\mathbbm{1})^j \,\cU_2 (t;s) \,\psi \|_{\mathcal{F}} \;&\leqslant\; \mathsf{K}_j \; e^{\mathsf{K}_j |t-s|} \| (\cN +\mathbbm{1})^j \psi \|_{\mathcal{F}}\,, \label{eq:boundgrowth2}
 \end{align}
 valid for every $t,s\in\mathbb{R}$, $2j\in\mathbb{N}$, $\psi\in\mathcal{F}$, for some ($N$- and $t,s$-independent) constants $\mathsf{K}_j\equiv\mathsf{K}_j(\mathsf{q},\|\phi\|_{\mathcal{Q}})$, where $\mathcal{F}$ is the underlying Fock space for this analysis ($\|\,\|_{\mathcal{F}}$ is its norm and $\langle\cdot,\cdot\rangle_{\mathcal{F}}$ is its scalar product), and $\mathcal{N}$ is the number operator on $\mathcal{F}$. It is important to stress that this derivation of \eqref{eq:boundgrowth1}-\eqref{eq:boundgrowth2} following \cite[Proposition 3.3 and Lemma 3.5]{RS-2007} shows that the constants $\mathsf{K}_j$ are \emph{independent} of the parameter $\nu$. 
 
 \medskip

   \textbf{Reduced fluctuation dynamics.} Further intermediate results for the above adaptation of \cite[Proposition 2.1]{Chen-Lee-Schlein-2011} are needed, consisting of the adaptation of suitable controls on the generator $\mathcal{L}_2$ of the reduced fluctuation dynamics. To this aim, one considers the kinetic energy operator $\mathcal{K}$ and the external potential operator $\mathcal{W}$ on Fock space defined, respectively, by
   \begin{equation}\label{eq:defKW}
\cK \,:=\, \int \ud x \, \nabla_x a_x^* \, \nabla_x a_x
\qquad
\text{and}
\qquad
\cW \,:=\, \int \ud x \, U(x)\,a_x^* a_x\,.
\end{equation}
\begin{itemize}
 \item First, the proof of \cite[Lemma 6.1]{Chen-Lee-Schlein-2011} can be repeated, now with $\mathcal{L}_2$ defined by \eqref{eq:L2}, thereby replacing the original role of $\mathcal{K}$ with the present $\nu  \,\cK + \nu  \, \cW$, and replacing also the original assumption $V^2\lesssim\mathbbm{1}-\Delta_x$ with the present \eqref{eq:Vcond}: this yields the operator inequalities
   \begin{equation}\label{lm:7}
-\mathsf{C}^{(\mathsf{q})} (\N+\mathbbm{1}) \;\leqslant\; \cL_2(t) - \nu  \,\cK - \nu  \, \cW \;\leqslant\;  \mathsf{C}^{(\mathsf{q})} (\N+\mathbbm{1})\,,
\end{equation} 
 valid uniformly in $t\in\mathbb{R}$, for some constant $\mathsf{C}^{(\mathsf{q})}$
 depending on $\mathsf{q}$ owing to \eqref{eq:unifbound}.
 
 \item Furthermore, the proof of \cite[Lemma 6.2]{Chen-Lee-Schlein-2011} can be repeated as well, again with the present definition \eqref{eq:L2} of $\mathcal{L}_2$, and one obtains the bounds
 \begin{equation}\label{eq:U2L2U2}
 \begin{split}
  \big| \big\langle \U_2 (t;s) \psi\,,\, \cL_2 (t)\, \U_2 (t;s) \psi \big\rangle_{\mathcal{F}} \big| 
\;&\leqslant\; 
e^{\frac{\mathsf{G}}\nu |t-s|} \big\langle \, \psi\,, \cL_2(s) \, \psi \big\rangle_{\mathcal{F}} \\
&\qquad +\frac{\mathsf{G} }\nu  \frac{\,(\nu  +\mathsf{C}^{(\mathsf{q})})\,\mathsf{K}_{\frac{1}{2}}^2\,}{\mathsf{K}_{\frac{1}{2}} - \mathsf{G}/\nu  }\, e^{2\mathsf{K_{\frac{1}{2}}}|t-s|}
\big\langle \psi\,, (\cN + \mathbbm{1}) \, \psi \big\rangle_{\mathcal{F}}
 \end{split}
\end{equation}
 (where $\mathsf{C}^{(\mathsf{q})}$
 and $\mathsf{K}_{\frac{1}{2}}$ are the constants provided, respectively, by \eqref{lm:7} and \eqref{eq:boundgrowth2}), valid for any $t,s\in\mathbb{R}$ and $\psi\in\mathcal{F}$, for some $t,s$-independent constant 
 $\mathsf{G}$ whose dependence is $\mathsf{G}\equiv\mathsf{G}(\mathsf{q},\|\phi^0\|_{\mathcal{Q}})$ owing to \eqref{eq:unifbound} and \eqref{eq:Vcond}. Let us sketch the tricky emergence of $\nu$ in \eqref{eq:U2L2U2}, as an adaptation from \cite[Lemma 6.2]{Chen-Lee-Schlein-2011}. The initial part of the analysis is the same as in the original version and yields
 \[
\left|
\frac{\ud}{\ud t} \big\langle \U_2 (t;s) \psi\,,\, \cL_2 (t)\, \U_2 (t;s) \psi \big\rangle_{\mathcal{F}} 
\right|\;\leqslant\; \mathsf{G}\, \langle \U_2 (t;s) \psi, (\cK + \cN + \mathbbm{1}) \U_2 (t;s) \psi\rangle_{\mathcal{F}}
\]
  for some $\mathsf{G}\equiv\mathsf{G}(\mathsf{q},\|\phi^0\|_{\mathcal{Q}})$ emerging from the application of \eqref{eq:unifbound} and \eqref{eq:Vcond}. 
    Applying \eqref{lm:7} one obtains
  \begin{align*}
&\left|
\frac{\ud}{\ud t} \big\langle \U_2 (t;s) \psi\,,\, \cL_2 (t)\, \U_2 (t;s) \psi \big\rangle_{\mathcal{F}} 
\right|\;\leqslant\; \mathsf{G}\,
\Big\langle \U_2 (t;s) \psi\,, \Big(\frac{\nu  \cK}{\nu } + \cN + \mathbbm{1}\Big) \U_2 (t;s) \psi\Big\rangle_{\!\mathcal{F}}\\
&\qquad\qquad\leqslant\; \mathsf{G}\,
\Big\langle \U_2 (t;s) \psi\,, \,\frac{\,\cL_2(t)-\nu  \,\cW+(\nu  +\mathsf{C}^{(\mathsf{q})})(\cN+\mathbbm{1})}\nu \;\U_2 (t;s) \psi\Big\rangle_{\!\mathcal{F}}\,.
\end{align*}
 Using the positivity of $U$, and hence of $\mathcal{W}$, 
\[
\left|
\frac{\ud}{\ud t}  \big\langle \,\U_2 (t;s) \psi\,, \cL_2(t) \,\U_2 (t;s) \psi \big\rangle_{\mathcal{F}}
\right|\;
\leqslant\;
\frac{\mathsf{G}}\nu 
 \,\big\langle \,\U_2 (t;s) \psi\,, \big(\cL_2(t) + (\nu +\mathsf{C}^{(\mathsf{q})})(\cN + \mathbbm{1}) \big) \,\U_2 (t;s) \psi \big\rangle_{\mathcal{F}}\,.
\]
Then a Gr\"onwall-type argument
and \eqref{eq:boundgrowth2} finally yield \eqref{eq:U2L2U2}.
\end{itemize}

 \medskip
 
\textbf{Comparison of dynamics.} The next technical tool which the original proof of \cite[Proposition 2.1]{Chen-Lee-Schlein-2011} is based upon, and which need be adapted to the present setting in order to validate the above derivation of the bound \eqref{trace norm bound}, is a convenient estimate of the difference between $\cU$ and $\cU_2$. The net result, mirroring \cite[Proposition 6.1]{Chen-Lee-Schlein-2011}, is
\begin{equation}\label{eq:comparisondynamics}
\begin{split}
 &\big\| (\cN+\mathbbm{1})^j \, \left( \cU (t;s) - \cU_2 (t;s)\right) \psi \big\|_{\mathcal{F}}\\
 &\;\leqslant\; 
 \frac{\,\mathsf{M}_j\,e^{\mathsf{M}_j t}}{\sqrt{N}}\, \Big(\big\langle\psi\,,\big(\,\cK+ \,\cW+2 {\textstyle\frac{\,\mathsf{C}^{(\mathsf{q})}+\nu\,}{\nu}} (\cN+\mathbbm{1})\big)\,\psi\big\rangle_{\mathcal{F}}^{\!\frac{1}{2}}+\big\| (\cN+\mathbbm{1})^{6(2j+3)} \psi \big\|_{\mathcal{F}}\Big)\,,
\end{split}
\end{equation}
 valid for every $t,s\in\mathbb{R}$, $j\in\mathbb{N}$, $\psi\in\mathcal{F}$ for some ($N$- and $t$-independent) constants $\mathsf{M}_j\equiv\mathsf{M}_j(\mathsf{q},\|\phi^0\|_{\mathcal{Q}}, \nu)$. Moreover, each $\mathsf{M}_j$ is uniformly bounded in $\nu$ for $\nu\geqslant\nu_0$, with arbitrary $\nu_0>0$.
 This is obtained by re-doing the steps of the proof of \cite[Proposition 6.1]{Chen-Lee-Schlein-2011} with $\cU$ and $\cU_2$ given now by \eqref{eq:cL}-\eqref{eq:L2} and \eqref{eq:U}-\eqref{eq:U2}, and with the assumption \eqref{eq:Vcond}. Let us provide details on that.

 For concreteness, fix $t\geqslant 0$ and $s=0$ (all other cases are treated analogously). From \eqref{eq:U}-\eqref{eq:U2} one writes
\[  \cU (t;0) - \cU_2 (t;0) \;=\; \int_0^t \ud \tau \: \cU (t;\tau) ( \cL (\tau) - \cL_2 (\tau) )\, \cU_2\, (\tau;0)\,. \]
 Using this, re-writing $\cL(\tau)-\cL_2(\tau)=\cL_3(\tau) + \cL_4 $ with
\begin{equation}\label{eq:defL3L4}
\begin{split} \cL_3 (\tau) \;&:=\; \frac{1}{\sqrt{N}} \int \ud x \,\ud y \, \wt V(x-y) \, a_x^* (a_y^* \wt \phi_{\tau} (y) + a_y \overline{\wt \phi_{\tau} (y)}) a_x\,, \\   \cL_4 \;&:=\; \frac{1}{N} \int \ud x\, \ud y \, \wt V(x-y) \, a_x^* a_y^* a_y a_x  \, , \end{split} \end{equation}
 and applying the bounds \eqref{eq:boundgrowth1}-\eqref{eq:boundgrowth2}, we find
  \begin{equation}\label{eq:L3+L4} 
  \begin{split} 
 \big\| (\cN+\mathbbm{1})^j \, ( \cU (t;0) - \cU_2 (t;0) ) \psi \|_{\mathcal{F}} \;&\leqslant \;  \int_0^t \ud \tau \, \big\| (\cN+\mathbbm{1})^j \,\cU (t;\tau) \left( \cL (\tau) - \cL_2 (\tau) \right) \cU_2 (\tau;0) \,\psi \big\|_{\mathcal{F}} \\ 
 &\leqslant \; \mathsf{K}_j\int_0^t \ud \tau \, e^{\mathsf{K}_j (t-\tau)} 
 \big\| (\cN+\mathbbm{1})^{2j+1} \cL_3 (\tau) \,  \cU_2 (\tau;0) \psi \big\|_{\mathcal{F}} \\ 
 & \qquad + \mathsf{K}_j\int_0^t \ud \tau \, e^{\mathsf{K}_j (t-\tau)} \big\| (\cN+\mathbbm{1})^{2j+1} \cL_4 \, \cU_2 (\tau;0) \psi \big\|_{\mathcal{F}}\,.
\end{split} 
\end{equation}
  To estimate the first summand in the r.h.s.~above one re-proves \cite[Lemma 6.3]{Chen-Lee-Schlein-2011}, where now condition \eqref{eq:Vcond} replaces the original assumption $V^2\lesssim\mathbbm{1}-\Delta_x$, and the bound \eqref{eq:unifbound} is used, thereby obtaining
  \[
   \big\| (\cN+\mathbbm{1})^j \cL_3 (\tau) \psi \big\|_{\mathcal{F}} \;\leqslant\; \frac{\mathsf{Q}}{\sqrt{N}}\, \big\| (\cN+\mathbbm{1})^{j+\frac{3}{2}} \psi \big\|_{\mathcal{F}}
  \]
  uniformly in $\tau$ and $j\in\mathbb{N}$, for some ($N$- and $j$-independent) constant $\mathsf{Q}\equiv\mathsf{Q}(\mathsf{q},\|\phi^0\|_{\mathcal{Q}})$. This and the unitarity of $\cU_2 (\tau;0)$ then yield
  \begin{equation}\label{eq:int1toplug}
   \int_0^t \ud \tau \, e^{\mathsf{K}_j (t-\tau)} 
 \big\| (\cN+\mathbbm{1})^{2j+1} \cL_3 (\tau) \,  \cU_2 (\tau;0) \psi \big\|_{\mathcal{F}}\;\leqslant\;\frac{\,\mathsf{Q}\,e^{\mathsf{K}_j t}}{\sqrt{N}}\, \big\| (\cN+\mathbbm{1})^{j+\frac{5}{2}} \psi \big\|_{\mathcal{F}}\,.
  \end{equation}
  The estimate of the second summand in the r.h.s.~of \eqref{eq:L3+L4} requires a control of
  \begin{equation}\label{eq:takeroot}
    \big\| (\cN+\mathbbm{1})^{2j+1} \cL_4 \, \cU_2 (\tau;0) \psi \big\|_{\mathcal{F}}\;=\;\big\langle\, \cU_2 (\tau;0) \psi\,,\,  (\cN+\mathbbm{1})^{2j+1} \cL^2_4\, (\cN+\mathbbm{1})^{2j+1} \,\cU_2 (\tau;0) \psi\big\rangle_{\mathcal{F}}^{\frac{1}{2}}\,.
  \end{equation}
  To this aim, one estimates first the restriction of $(\cN+\mathbbm{1})^{2j+1} \cL^2_4\, (\cN+\mathbbm{1})^{2j+1} $ on the $n$-particle sector $\mathcal{F}_n$ of the Fock space $\mathcal{F}$: repeating this step from the proof of \cite[Proposition 6.1]{Chen-Lee-Schlein-2011} we find  
  \[ \begin{split}
(\cN+\mathbbm{1})^{2j+1} &\cL^2_4\, (\cN+\mathbbm{1})^{2j+1}\Big|_{\mathcal{F}_n} \;= \; \frac{(n+1)^{4j+2}}{N^2} \bigg( \sum_{1\leqslant\ell<r\leqslant n}\wt{V}(x_\ell-x_r) \bigg)^{\!2}\;\leqslant\;\frac{(n+1)^{4j+4}}{N^2}  \!\!\!\sum_{1\leqslant\ell<r\leqslant n}\wt{V}^2(x_\ell-x_r)  \\ 
 &\leqslant\; \frac{\,{\bf 1} \big(n+1 \leqslant N^{\frac{1}{4j+5}}\big)}{N^2} (n+1)^{4j+5} \,\mathsf{q}\sum_{\ell=1}^n (1 -\Delta_{x_\ell} + U(x_\ell)) + N^4(n+1)^{4j+6}\,{\bf 1} \big(n+1 \geqslant N^{\frac{1}{4j+5}}\big)  \\ 
 & \leqslant\;\left.\left( \frac{\,\mathsf{C}^{(\mathsf{q})}}{N} ( \cN + \cK + \cW) + N^4(\cN+\mathbbm{1})^{4j+6}\,{\bf 1} \big(\cN+\mathbbm{1} \geqslant N^{\frac{1}{4j+5}}\big)\right) \right|_{\cF^{(n)}}\,, 
 \end{split} 
\]
 having now applied \eqref{eq:Vcond} (whence the $\mathsf{q}$-dependence) and \eqref{eq:wtV}, tacitly re-naming $\mathsf{C}^{(\mathsf{q})}$ as the largest value between the constant needed here and the constant $\mathsf{C}^{(\mathsf{q})}$ from \eqref{lm:7}. This implies, by means of \eqref{lm:7},
   \[ \begin{split}
(\cN+\mathbbm{1})^{2j+1} &\cL^2_4\, (\cN+\mathbbm{1})^{2j+1}\\ 
 & \leqslant\; \frac{\,\mathsf{C}^{(\mathsf{q})}}{N} \frac{\,\cL_2 (\tau)+ \mathsf{C}^{(\mathsf{q})}(\cN+\mathbbm{1})}\nu  + N^4(\cN+\mathbbm{1})^{4j+6}\,{\bf 1} \big(\cN+\mathbbm{1} \geqslant N^{\frac{1}{4j+5}}\big) \,, 
 \end{split} 
\]
  whence also
  \begin{equation} \label{eq:H2+N}
   \begin{split}
    \big\langle\, \cU_2 (\tau;0) \psi\,,\,&  (\cN+\mathbbm{1})^{2j+1} \cL^2_4\, (\cN+\mathbbm{1})^{2j+1} \,\cU_2 (\tau;0) \psi\big\rangle_{\mathcal{F}} \\
    &\leqslant\;\frac{\,\mathsf{C}^{(\mathsf{q})}}{N}\Big\langle\, \cU_2 (\tau;0) \psi\,,\,\frac{\,\cL_2 (\tau)+ \mathsf{C}^{(\mathsf{q})}(\cN+\mathbbm{1})}\nu  \,\cU_2 (\tau;0) \psi\Big\rangle_{\!\mathcal{F}} \\
    &\qquad + N^4\big\langle\, \cU_2 (\tau;0) \psi\,,\,(\cN+\mathbbm{1})^{4j+6}\,{\bf 1} \big(\cN+\mathbbm{1} \geqslant N^{\frac{1}{4j+5}}\big)\,\cU_2 (\tau;0) \psi\big\rangle_{\mathcal{F}}\,.
   \end{split}
  \end{equation}
  Concerning the first summand in the r.h.s.~of \eqref{eq:H2+N}, one splits 
   \begin{equation*}
    \Big\langle\, \cU_2 (\tau;0) \psi\,,\,\frac{\,\cL_2 (\tau)+ \mathsf{C}^{(\mathsf{q})}(\cN+\mathbbm{1})}\nu  \,\cU_2 (\tau;0) \psi\Big\rangle_{\!\mathcal{F}}\;=\;(\mathrm{I})+(\mathrm{II})
   \end{equation*}
  with
  \[
   \begin{split}
    (\mathrm{I})\;&:=\; \nu ^{-1}\big\langle\, \cU_2 (\tau;0) \psi\,,\,\cL_2 (\tau) \,\cU_2 (\tau;0) \psi\big\rangle_{\!\mathcal{F}}\,, \\
    (\mathrm{II})\;&:=\; \nu ^{-1}\mathsf{C}^{(\mathsf{q})}\big\langle\, \cU_2 (\tau;0) \psi\,,\,(\cN+\mathbbm{1}) \,\cU_2 (\tau;0) \psi\big\rangle_{\!\mathcal{F}}\,.
   \end{split}
  \] 
  For $(\mathrm{I})$, \eqref{lm:7} and \eqref{eq:U2L2U2} give
  \begin{equation*}
  \begin{split}
  (\mathrm{I})&\;\leqslant\;
  e^{\frac{\mathsf{G}}\nu \,\tau} \langle \psi,  \big(\,\cK +  \,\cW + \nu ^{-1}\mathsf{C}^{(\mathsf{q})}(\cN+\mathbbm{1})\big) \psi \rangle_{\mathcal{F}}+\frac{\mathsf{G} }{\:\nu ^2} \frac{\,(\nu +\mathsf{C}^{(\mathsf{q})})\,\mathsf{K}_{\frac{1}{2}}^2\,}{\mathsf{K}_{\frac{1}{2}} - \mathsf{G}/\nu }\, e^{2\mathsf{K_{\frac{1}{2}}}\,\tau}
\big\langle \psi\,, (\cN + \mathbbm{1}) \, \psi \big\rangle_{\mathcal{F}}\,\\
&
\;\leqslant\;
\bigg(1+\frac{\mathsf{G} }{\:\nu ^2} \frac{\,(\nu +\mathsf{C}^{(\mathsf{q})})\,\mathsf{K}_{\frac{1}{2}}^2\,}{\mathsf{K}_{\frac{1}{2}} - \mathsf{G}/\nu }\bigg)\, e^{\mathsf{(\frac{G}\nu +2 K_{\frac{1}{2}}})\,\tau}
\big\langle \psi\,, 
\big(\,\cK +  \,\cW + \nu ^{-1}(\mathsf{C}^{(\mathsf{q})}+\nu )(\cN+\mathbbm{1})\big)
\, \psi \big\rangle_{\mathcal{F}}\,
.
  \end{split}
  \end{equation*}
For $(\mathrm{II})$, \eqref{eq:boundgrowth2} gives
\begin{equation*}
(\mathrm{II})\; 
\leqslant\;
\nu ^{-1}\mathsf{C}^{(\mathsf{q})}\, 
\mathsf{K}_{\frac12}^2\, e^{2\mathsf{K}_{\frac12} \tau}
\big\langle\,\psi\,,\,(\cN+\mathbbm{1}) \, \psi\big\rangle_{\!\mathcal{F}}\,.
\end{equation*}
 Thus,
   \begin{equation}\label{eq:firstRHStoplug}
   \begin{split}
     &\Big\langle\, \cU_2 (\tau;0) \psi\,,\,\frac{\,\cL_2 (\tau)+ \mathsf{C}^{(\mathsf{q})}(\cN+\mathbbm{1})}\nu  \,\cU_2 (\tau;0) \psi\Big\rangle_{\!\mathcal{F}}\;=\;(\mathrm{I})+(\mathrm{II}) \\
     &\quad\leqslant\;\bigg(1+\frac{\mathsf{G}}{\:\nu ^{2}}\frac{\,(\nu +\mathsf{C}^{(\mathsf{q})})\,\mathsf{K}_{\frac{1}{2}}^{2}\,}{\mathsf{K}_{\frac{1}{2}}-\mathsf{G}/\nu }+\frac{\,\mathsf{C}^{(\mathsf{q})}}{\nu}\bigg)\,e^{\mathsf{(\frac{G}\nu +2K_{\frac{1}{2}}})\,\tau}\,\big\langle\psi\,,\big(\,\cK+ \,\cW+2\nu ^{-1}(\mathsf{C}^{(\mathsf{q})}+\nu )(\cN+\mathbbm{1})\big)\,\psi\big\rangle_{\mathcal{F}}\,.
   \end{split}
   \end{equation}  
   On the other hand, the obvious inequality $\mathbf{1}(\lambda\geqslant 1)\leqslant \lambda^{5(4j+5)}$ and the bound \eqref{eq:boundgrowth2} yield
 \begin{equation}\label{eq:secondRHStoplug}
  N^4\big\langle\, \cU_2 (\tau;0) \psi\,,\,(\cN+\mathbbm{1})^{4j+6}\,{\bf 1} \big(\cN+\mathbbm{1} \geqslant N^{\frac{1}{4j+5}}\big)\,\cU_2 (\tau;0) \psi\big\rangle_{\mathcal{F}}\;\leqslant\;\frac{\,\mathsf{K}_{6(2j+3)}^2 e^{\tau\,\mathsf{K}_{6(2j+3)}}}{N}\,\big\| (\cN+\mathbbm{1})^{6(2j+3)}\psi\big\|^2_{\mathcal{F}}\,.
 \end{equation}
  Plugging \eqref{eq:firstRHStoplug} and \eqref{eq:secondRHStoplug} into \eqref{eq:H2+N} yields, in view of \eqref{eq:takeroot},  
  \begin{equation}\label{eq:int2toplug}
  \begin{split}
&\int_{0}^{t}\ud\tau\,e^{\mathsf{K}_{j}(t-\tau)}\big\|(\cN+\mathbbm{1})^{2j+1}\cL_{4}\,\cU_{2}(\tau;0)\psi\big\|_{\mathcal{F}}\\
&\;\leqslant\;e^{\mathsf{K}_{j}t}\sqrt{\frac{\mathsf{C}^{(\mathsf{q})}}{N}}\int_{0}^{t}\ud\tau\,\sqrt{(\mathrm{I})+(\mathrm{II})\,}+e^{\mathsf{K}_{j}t}\,\frac{\mathsf{K}_{6(2j+3)}}{\sqrt{N}}\,\int_{0}^{t}\ud\tau\,e^{\frac{1}{2}\tau\,\mathsf{K}_{6(2j+3)}}\big\|(\cN+\mathbbm{1})^{6(2j+3)}\psi\big\|_{\mathcal{F}}\\
&\;\leqslant\;\sqrt{\frac{\mathsf{C}^{(\mathsf{q})}}{N}}\bigg(1+\frac{\mathsf{G}}{\:\nu ^{2}}\frac{\,(\nu +\mathsf{C}^{(\mathsf{q})})\,\mathsf{K}_{\frac{1}{2}}^{2}\,}{\mathsf{K}_{\frac{1}{2}}-\mathsf{G}/\nu }+\frac{\,\mathsf{C}^{(\mathsf{q})}}{\nu}\bigg)^{\!\frac{1}{2}}\,\frac{\,e^{\mathsf{(\frac{G}{2\nu} +K_{\frac{1}{2}}}+\mathsf{K}_j)\,t}}{\mathsf{\,\frac{G}{2\nu} +K_{\frac{1}{2}}}}\big\langle\psi\,,\big(\,\cK+ \,\cW+2\nu ^{-1}(\mathsf{C}^{(\mathsf{q})}+\nu )(\cN+\mathbbm{1})\big)\,\psi\big\rangle_{\mathcal{F}}^{\frac{1}{2}}\\
&\quad+\frac{2}{\sqrt{N\,}\,}\,e^{\frac{1}{2}t\,(\mathsf{K}_{6(2j+3)}+\mathsf{K}_j)}\,\big\|(\cN+\mathbbm{1})^{6(2j+3)}\psi\big\|_{\mathcal{F}}\\
&\leqslant\;            \frac{1}{\sqrt{N\,}\,}\Bigg(\bigg(1+\frac{\mathsf{G}}{\:\nu ^{2}}\frac{\,(\nu +\mathsf{C}^{(\mathsf{q})})\,\mathsf{K}_{\frac{1}{2}}^{2}\,}{\mathsf{K}_{\frac{1}{2}}-\mathsf{G}/\nu }+\frac{\,\mathsf{C}^{(\mathsf{q})}}{\nu}\bigg)^{\!\frac{1}{2}}+2\Bigg)\,e^{(\frac{\mathsf{G}}{2\nu} +\mathsf{K}_{\frac{1}{2}}+\frac{1}{2}\mathsf{K}_{6(2j+3)}+\mathsf{K}_j)\,t}\,\times\\
&\qquad\qquad\times \Big(\big\langle\psi\,,\big(\,\cK+ \,\cW+2\nu ^{-1}(\mathsf{C}^{(\mathsf{q})}+\nu )(\cN+\mathbbm{1})\big)\,\psi\big\rangle_{\mathcal{F}}^{1/2}+\big\langle\psi\,,(\cN+\mathbbm{1})^{12(2j+3)}\,\psi\big\rangle_{\mathcal{F}}^{1/2}\Big)\,.
   \end{split}
  \end{equation}
  In turn, plugging \eqref{eq:int1toplug} and \eqref{eq:int2toplug} into \eqref{eq:L3+L4}, and setting
  \begin{equation}
   \mathsf{M}_j\;:=\;\Bigg(\bigg(1+\frac{\mathsf{G}}{\:\nu ^{2}}\frac{\,(\nu +\mathsf{C}^{(\mathsf{q})})\,\mathsf{K}_{\frac{1}{2}}^{2}\,}{\mathsf{K}_{\frac{1}{2}}-\mathsf{G}/\nu }+\frac{\,\mathsf{C}^{(\mathsf{q})}}{\nu}\bigg)^{\!\frac{1}{2}}+2\Bigg)+\Big(\frac{\mathsf{G}}{2\nu} +\mathsf{K}_{\frac{1}{2}}+\frac{1}{2}\mathsf{K}_{6(2j+3)}+\mathsf{K}_j\Big)
  \end{equation}
  finally produces \eqref{eq:comparisondynamics}. Observe that $\mathsf{M}_j$ has a finite limit as $\nu\to +\infty$ and therefore it is indeed uniformly bounded in $\nu$ for $\nu$ larger than any fixed threshold $\nu_0>0$.

   \medskip
   
   \textbf{Combining the preparations together for the bound \eqref{trace norm bound}.} One finally obtains \eqref{trace norm bound} following, with the obvious adaptations due to the modified estimates above, the scheme of \cite[Section 4]{Chen-Lee-Schlein-2011}. According to a customary line of reasoning in this context, one computes the trace norm in the l.h.s.~of \eqref{trace norm bound} by duality with compact operators. For concreteness when $k=1$ (the case of generic $k$ is completely analogous), and for any compact hermitian $J$ on $L^2(\mathbb{R}^3)$, acting with an integral kernel $J(x;y)$, one obtains 
    \begin{equation}\label{eq:diff1}
\begin{split}
\tr \, J \big( \wt\gamma^{(1)}_{N,\nu,t}  - |\wt \phi_t \rangle \langle \wt \phi_t| \big) \;&=\;\iint_{\mathbb{R}^3\times\mathbb{R}^3} \mathrm{d}x \,\mathrm{d}y \, J(x;y) \left(\wt \gamma^{(1)}_{N,\nu,t} (y;x) - \wt \phi_t (y) \, \overline{\wt \phi_t (x)}\, \right)  \\
&=\; \frac{d_N}{N} \bigg\langle \frac{a^* (\widetilde{\phi}_t)^N}{\sqrt{N!}}\, \Omega \,, W (\sqrt{N} \widetilde{\phi}_t) \,\cU^* (t;0) \,\,\ud\Gamma (J) \,\cU(t;0) \,\Omega\bigg\rangle_{\!\mathcal{F}} \\ 
&\qquad + \frac{\:d_N}{\sqrt{N\,}\,} \bigg\langle \frac{a^* (\widetilde{\phi}_t)^N}{\sqrt{N!}} \,\Omega \,, W (\sqrt{N} \widetilde{\phi}_t)\, \cU^* (t;0)\, \Phi (J \widetilde{\phi}_t) \,\cU(t;0)\, \Omega\bigg\rangle_{\!\mathcal{F}}
\end{split}
\end{equation}
 precisely as \cite[Eq.~(4.9)]{Chen-Lee-Schlein-2011}, with the usual meaning of symbols for the Fock space vacuum vector $\Omega$, the second quantisation mapping $\ud\Gamma(\cdot)$, the creation operator $a^*(\cdot)$, the Weyl operator $W(\cdot)$, the Segal operator $\Phi(\cdot)$ (with the normalisation convention $\Phi(f)=\overline{a^*(f)+a(f)}$\,,  $f\in L^2(\mathbb{R}^3)$), with $\mathcal{U}(t;0)$ now given by \eqref{eq:cL} and \eqref{eq:U}, and with 
 \begin{equation}
  d_N\;:=\;\frac{\sqrt{N!\,}}{\,(N/e)^{\frac{N}{2}}}\;\lesssim\;N^{\frac{1}{4}}\,.
 \end{equation}
 Then, as customary, the first summand of the r.h.s.~in \eqref{eq:diff1} is estimated, by inserting $\mathbbm{1}=(\mathcal{N}+\mathbbm{1})^{-\frac{1}{2}}(\mathcal{N}+\mathbbm{1})^{\frac{1}{2}}$ and by a Schwarz inequality, that leads to estimating the quantity
 \[
  \big\|(\mathcal{N}+\mathbbm{1})^{\frac{1}{2}}\cU^* (t;0) \,\,\ud\Gamma (J) \,\cU(t;0) \,\Omega\big\|_{\mathcal{F}}\,,
 \]
 which is now done by means of \eqref{eq:boundgrowth1}.
 The second summand of the r.h.s.~in \eqref{eq:diff1} is treated by splitting
 \[
  \cU(t;0)\;=\;\big(\,\cU(t;0)-\cU_2(t;0)\big)+\cU_2(t;0)\,.
 \]
 The terms involving the difference $\cU(t;0)-\cU_2(t;0)$ with the comparison dynamics are suitably dealt with by means of the bounds \eqref{eq:comparisondynamics} (in particular the analogue of \cite[Proposition 6.2]{Chen-Lee-Schlein-2011} is re-obtained with the very same arguments therein, except for using now the inequality \eqref{lm:7} in the form 
 \[
\cK \;=\; \frac{\nu \,\cK}\nu \; \leqslant\; \frac{\,\cL_2(t) - \nu \,\cW + \mathsf{C}^{(\mathsf{q})}(\cN+\mathbbm{1})\,}\nu
\]
 and adapting the subsequent arguments accordingly, in the same spirit of the reasoning that yielded   \eqref{eq:comparisondynamics} here. This way, one reproduces the whole scheme leading to \eqref{trace norm bound} -- we refer to \cite[Section 4]{Chen-Lee-Schlein-2011} for details. In particular, having tracked down the dependence of all relevant constants on the parameters of interest, it is eventually possible to cast them into final constants $\mathsf{D}_k\equiv\mathsf{D}_k(\mathsf{q},\|\phi^0\|_{\mathcal{Q}}, \nu)>0$ appearing in \eqref{trace norm bound}, each of which is uniformly bounded in $\nu$ for $\nu\geqslant\nu_0$, with arbitrary $\nu_0>0$.

 \medskip
 
 The proof of  Theorem \ref{thm:CLSmain} is thus completed.
 \end{proof}

\vspace{2em}
\section*{Acknowledgements}
This project is partially supported by the Deutsche Forschungsgemeinschaft under the Excellence Strategy no.~EXC-2111-390814868 (J.~L.), by the  Italian National Institute for Higher Mathematics -- INdAM (A.~M.), and by the von Humboldt Foundation, Bonn (A.M.). J.~L.~and A.~M.~gratefully acknowledge the kind hospitality, respectively, of the Center for Advanced Studies CAS-LMU Munich, and of the Mathematical Institute at the Silesian University Opava, where the final part of this project was completed. Both authors are indebted with D.~Dimonte for insightful discussions on the subject, and in particular on the article \cite{DimFalcOlg21-fragmented}.

\vspace{2em}

\def\cprime{$'$}

\end{document}